\documentclass[3p]{elsarticle}
\usepackage{lineno}
\usepackage{hyperref}
\modulolinenumbers[5]
\usepackage[noend]{algpseudocode}
\usepackage[linesnumbered,ruled,noend,vlined]{algorithm2e}
\usepackage{caption}
\usepackage[singlelinecheck=on]{caption}
\usepackage{amsmath,amssymb,amsfonts}
\usepackage{amsthm}
\usepackage{booktabs}
\usepackage{float} 
\usepackage{flushend}
\usepackage{graphicx}
\usepackage{mathrsfs}
\usepackage{multirow}
\usepackage{accents}
\usepackage{ragged2e}
\usepackage{setspace}
\usepackage{subfigure}
\usepackage{tabularx}
\usepackage{textcomp}
\usepackage{url}
\usepackage{verbatim}
\usepackage{xcolor}
\usepackage{xspace}
\newcommand{\kcore}{\textit{k}-core\xspace}
\newcommand{\ktruss}{\textit{k}-truss\xspace}
\newcommand{\kbitruss}{\textit{k}-bitruss\xspace}
\newcommand{\beindex}{\texttt{BE}-\texttt{Index}\xspace}
\newcommand{\grating }{2D-index\xspace}
\newcommand{\gratings }{2D-indexes\xspace}
\newcommand{\decompnaive}{\texttt{decomp-naive}\xspace}
\newcommand{\decompopt}{\texttt{decomp-opt}\xspace}
\newcommand{\btf}{\,\mathbin{\resizebox{0.1in}{!}{\rotatebox[origin=c]{90}{$\Join$}}}}
\newcommand{\BTF}{\,\mathbin{\resizebox{0.15in}{!}{\rotatebox[origin=c]{90}{$\Join$}}}}
\newcommand{\shortname}{$(\alpha,\beta)_{\tau}$-core\xspace}
\newcommand{\shortnames}{$(\alpha,\beta)_{\tau}$-cores\xspace}
\newcommand{\longname}{$\tau$-strengthened $(\alpha,\beta)$-core\xspace}
\newcommand{\support}{sup \xspace}
\newcommand{\engage}{eng \xspace}
\newcommand{\Query}{DecompQuery}
\newcommand{\dmaxu}{d_{max}(U)\xspace}
\newcommand{\dmaxv}{d_{max}(L)\xspace}
\newcommand{\peelG}{T_{peel}(G)}
\newcommand{\izero}{I_{\alpha,\beta,\tau}}
\newcommand{\ione}{I_{\alpha,\beta}}
\newcommand{\itwo}{I_{\beta,\tau}}
\newcommand{\nb}{nb}
\newcommand{\ithree}{I_{\alpha,\tau}}

\newcommand{\abcore}{$(\alpha,\beta)$-core\xspace}
\newcommand{\ab}{(\alpha,\beta)\textnormal{-}core}
\newcommand{\abg}{(\alpha,\beta)_{\tau}\textnormal{-}core}
\newcommand{\btftime}{T_{BE}}
\newcommand{\AII}{$(\alpha,1)_{1}$-core\xspace}
\newcommand{\ABI}{$(\alpha,\beta)_{1}$-core\xspace}
\newcommand{\IBG}{$(1,\beta)_{\tau}$-core\xspace}
\newcommand{\AIG}{$(\alpha,1)_{\tau}$-core\xspace}

\newcommand{\aii}{(\alpha,1)_{1}\textnormal{-}core}
\newcommand{\abi}{(\alpha,\beta)_{1}\textnormal{-}core}
\newcommand{\ibg}{(1,\beta)_{\tau}\textnormal{-}core}
\newcommand{\aig}{(\alpha,1)_{\tau}\textnormal{-}core}
\newcommand{\ibi}{(1,\beta)_{1}\textnormal{-}core}
\newcommand{\qbs}{Q_{bs}}
\newcommand{\qzero}{Q_{\alpha,\beta,\tau}}
\newcommand{\qone}{Q_{\alpha,\beta}}
\newcommand{\qtwo}{Q_{\beta,\tau}}
\newcommand{\qthree}{Q_{\alpha,\tau}}

\newcommand{\compute} {Peeling\xspace}

\def\BibTeX{{\rm B\kern-.05em{\sc i\kern-.025em b}\kern-.08em
    T\kern-.1667em\lower.7ex\hbox{E}\kern-.125emX}}
\newtheorem{example}{Example}

\newtheorem{lemma}{Lemma}

\newtheorem{definition}{Definition}
\newcommand{\cate}{\mathbin{\resizebox{0.098in}{!}{\rotatebox[origin=c]{270}{$\ltimes$}}}}

\begin{document}

\begin{frontmatter}

\title{Exploring Cohesive Subgraphs with Vertex Engagement and Tie Strength in Bipartite Graphs}

\author{$^{1}$Yizhang He}
\ead{yizhang.he@unsw.edu.au}

\author{$^{1}$Kai Wang\corref{mycorrespondingauthor}}
\cortext[mycorrespondingauthor]{Corresponding author}
\ead{kai.wang@unsw.edu.au}

\author{$^1$Wenjie Zhang}
\ead{zhangw@cse.unsw.edu.au}

\author{$^1$Xuemin Lin}
\ead{lxue@cse.unsw.edu.au}

\author{$^2$Ying Zhang}
\ead{ying.zhang@uts.edu.au}

\fntext[f1]{This manuscript is the authors' original work and has not been published nor has it been submitted simultaneously elsewhere.}
\fntext[f2]{ All authors have checked the manuscript and have agreed to the submission.}

\address{$^{1}$School of Computer Science and Engineering, University of New South Wales, NSW 2033, Australia
\\ $^{2}$Centre for AI, University of Technology Sydney, NSW 2007, Australia}

\begin{abstract}
We propose a novel cohesive subgraph model called $\tau$-strengthened $(\alpha,\beta)$-core (denoted as $(\alpha,\beta)_{\tau}$-core), which is the first to consider both tie strength and vertex engagement on bipartite graphs. An edge is a strong tie if contained in at least $\tau$ butterflies ($2\times2$-bicliques). $(\alpha,\beta)_{\tau}$-core requires each vertex on the upper or lower level to have at least $\alpha$ or $\beta$ strong ties, given strength level $\tau$. To retrieve the vertices of $(\alpha,\beta)_{\tau}$-core optimally, we construct index $I_{\alpha,\beta,\tau}$ to store all $(\alpha,\beta)_{\tau}$-cores. Effective optimization techniques are proposed to improve index construction. To make our idea practical on large graphs, we propose 2D-indexes $I_{\alpha,\beta}, I_{\beta,\tau}$, and $I_{\alpha,\tau}$ that selectively store the vertices of $(\alpha,\beta)_{\tau}$-core for some $\alpha,\beta$, and $\tau$. The 2D-indexes are more space-efficient and require less construction time, each of which can support $(\alpha,\beta)_{\tau}$-core queries. As query efficiency depends on input parameters and the choice of 2D-index, we propose a learning-based hybrid computation paradigm by training a feed-forward neural network to predict the optimal choice of 2D-index that minimizes the query time. Extensive experiments show that ($1$) $(\alpha,\beta)_{\tau}$-core is an effective model capturing unique and important cohesive subgraphs; ($2$) the proposed techniques significantly improve the efficiency of index construction and query processing.
\end{abstract}


\end{frontmatter}


\section{Introduction}
Bipartite graphs are widely used to represent networks with two different groups of entities such as user-item networks \cite{wang2006unifying}, author-paper networks \cite{konect:DBLP}, and member-activity networks \cite{brunson2015triadic}.
In bipartite graphs, cohesive subgraph mining has numerous applications including fraudsters detection \cite{allahbakhsh2013collusion,beutel2013copycatch,liu2020efficient}, group recommendation \cite{ding2017efficient,ntoutsi2012fast} and discovering inter-corporate relations \cite{ornstein1982interlocking,palmer2002interlocking}.  

($\alpha$,$\beta$)-core and bitruss are two representative cohesive subgraph models in bipartite graphs extended from the unipartite \textit{k}-core \cite{seidman1983network} and \textit{k}-truss \cite{cohen2008trusses} models.
\abcore is the maximal subgraph of a bipartite graph \textit{$G$} such that the vertices on upper or lower layer have at least $\alpha$ or $\beta$ neighbors respectively. 
\abcore models vertex engagement as degrees and treats each edge equally, but ties (edges) in real networks have different strengths.
\kbitruss is the maximal subgraph where each edge is contained in at least $k$ butterflies (i.e. $2$x$2$-biclique), which can model the tie strength
\cite{sariyuce2018peeling,zou2016bitruss}.

\begin{figure}[htb]
\centering  
\includegraphics[width=0.50\textwidth]{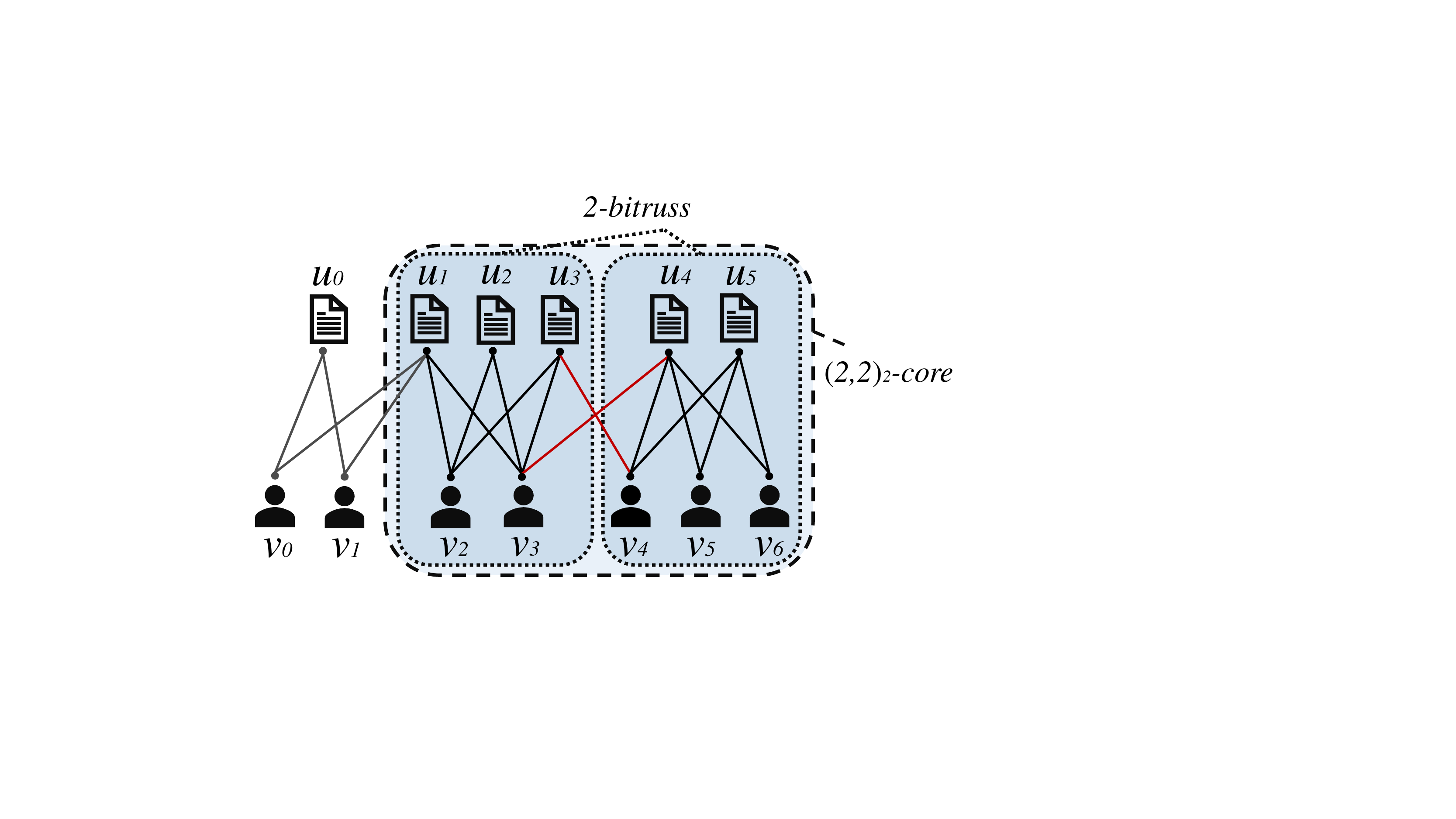}
\caption{Motivation example}
\label{fig:moltivation}
\end{figure}

In the author-paper network as shown in Figure \ref{fig:moltivation}, the graph is the $(\alpha,\beta)$-core ($\alpha$=$2$, $\beta$=$2$) and the light blue region is the \kbitruss ($k$=$2$). 
Without considering tie strength, $(\alpha,\beta)$-core blindly includes research groups of different levels of cohesiveness. We can see that $v_0$ and $v_1$ are not as closely connected as the rest authors. 
The \kbitruss model can exclude the relatively sparse subgraph containing $v_0$ and $v_1$, but it also deletes edges $(u_3,v_4)$ and $(u_4,v_3)$ when their incident vertices are present. This exposes the drawbacks of the \kbitruss model: ($1$) As \kbitruss only keeps strong ties, the weak ties between important vertices are missed.
In Fig \ref{fig:moltivation}, it fails to recognize the contributions of authors $v_3,v_4$ in papers $u_3,u_4$.
($2$) After removing weak ties, the tie strengths are modeled inaccurately. Edges $(u_3,v_3)$ and $(u_4,v_4)$ have more supporting butterflies ($u_3,u_4,v_3,v_4$ form a butterfly) than $(u_1,v_2)$, but their tie strengths are modeled as equal. 

In this paper, we study the efficient and scalable computation of \longname, 
which is the first cohesive subgraph model on bipartite graphs to consider both tie strength and vertex engagement. 
Given a bipartite graph $G$, we model the tie strength of each edge as the number of butterflies containing it. With a strength level $\tau$, we consider the edges with tie strength no less than $\tau$ to be \textit{strong ties}. 
The \textit{engagement} of a vertex is modeled as the number of strong ties it is incident to. 
Given engagement constraints $\alpha,\beta$ and strength level $\tau$, \shortname is the maximal subgraph of $G$ such that each upper or lower vertex in the subgraph has at least $\alpha$ or $\beta$ strong ties. The \shortname model is highly flexible and is able to capture unique structures.
For instance, in Figure \ref{fig:moltivation}, the subgraph induced by vertices $\{u_1, u_2, u_3, u_4, u_5, v_2, v_3, v_4, v_5, v_6\}$ is the $(2,2)_2$-core which cannot be found by \abcore or \kbitruss for any $\alpha,\beta$ or $k$.
Also, as shown in Figure \ref{fig:moltivation}, \abcore can preserve the weak ties if the incident vertices are present (e.g., the red edges are preserved due to $u_3,u_4,v_3$ and $v_4$), which better resembles reality.
The flexibility of the \shortname model is also evaluated in another experiment conducted on dataset \texttt{DBpedia-producer}. 
Figure \ref{fig:profile} shows the subgraphs of different densities found by \shortname and \abcore, where density is the ratio between the number of existing edges and the number of all possible edges \cite{sariyuce2018peeling}. 
$165$ subgraphs with a density greater than $0.2$ are found by \shortname while only $9$ such subgraphs are found by \abcore.

\begin{figure}[tbh]
\centering  
\subfigure[subgraphs found by \shortname]{
\label{fig.abt.profile}
\includegraphics[width=0.38\textwidth]{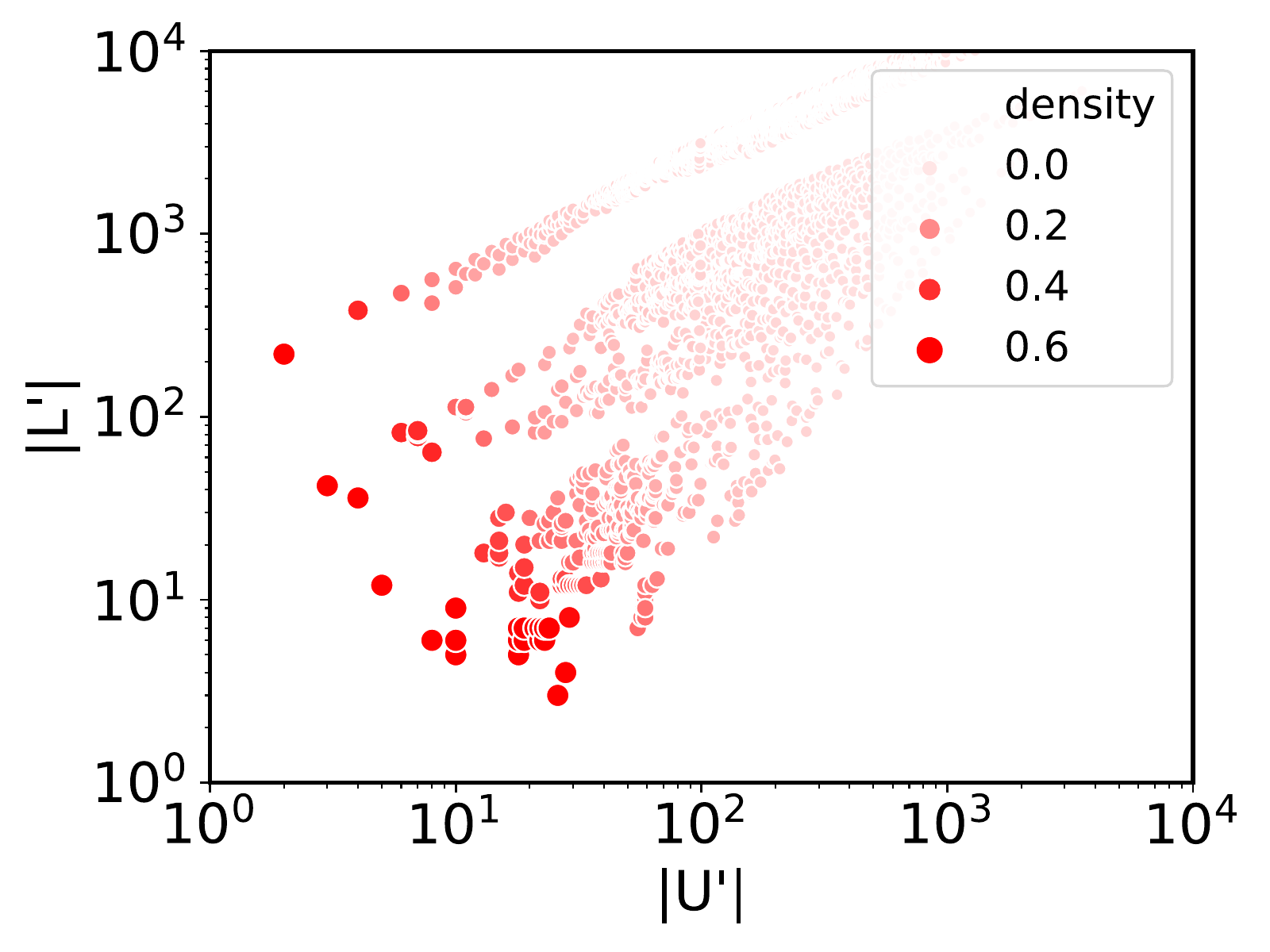}}
\subfigure[subgraphs found by \abcore]{
\label{fig.ab.profile}
\includegraphics[width=0.38\textwidth]{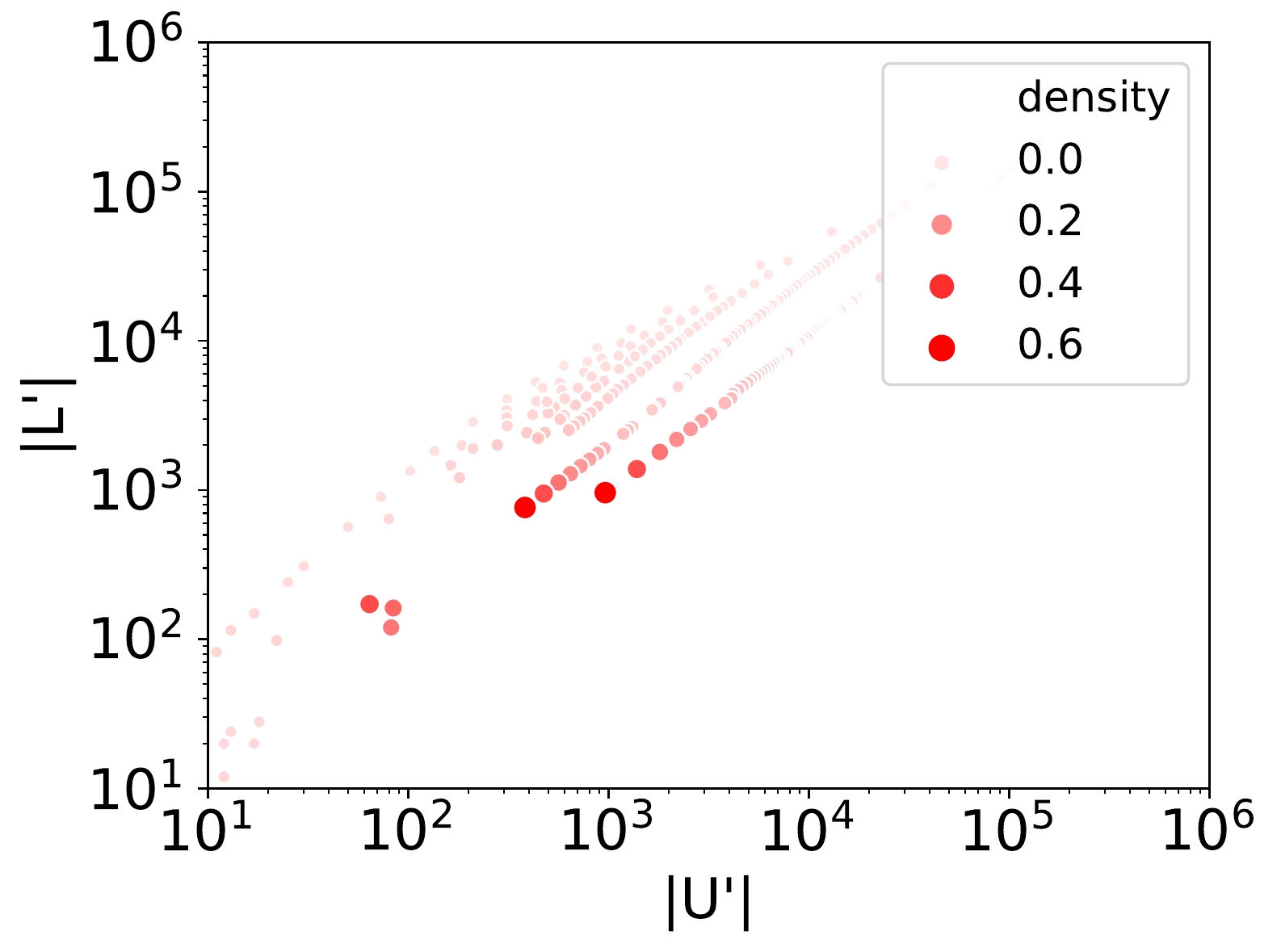}}
\caption{Dense subgraphs in \texttt{DBpedia-producer}.}
\label{fig:profile}
\end{figure}

\noindent
{\bf Applications.} The \longname model has many applications. We list some of them below.

\noindent
$\bullet$ \textit{Identify nested communities.}
On Internet forums like Reddit, Quora, and StackOverflow, users hold conversations on topics that interest them. The users and the topics form a bipartite network. 
In these networks, communities naturally exist and are nested.
For instance, Reddit displays a list of top communities like ``News", ``Gaming" and ``Sports" on the front page. 
``Sports" community contains many sub-communities including ``Cricket", ``Bicycling" and ``Golf". 
The edges in sub-communities have higher tie strength because users and topics within them are more closely connected. 
By increasing strength level $\tau$, \shortname captures the subgraphs forming a hierarchy, which can model nested communities on bipartite networks. 

\noindent
$\bullet$ \textit{Group similar users and items.}
In online shopping platforms like Amazon, eBay and Alibaba, users and items form a bipartite graph, where each edge indicates a purchasing record. 
Such a network consists of many closely connected communities, where some items are repeatedly bought by the same group of users (i.e., the target market). 
Examples of such communities include gym-goers and gym attires, students and stationery, diabetic patients and no-sugar foods, etc. 
Within one community, items are considered more similar and users tend to be alike due to their common shopping habits. 
As the edges between these users and items have high tie strength (butterfly support), we can use \shortname to find these communities and group similar users or items together.

\noindent
\textbf{Challenges.}
To obtain the \shortname from the input graph, we can first compute the support of edges and the engagement of vertices and then iteratively delete the vertices not meeting the engagement constraints. 
When $\alpha$, $\beta$, $\tau$ are large, \shortname is small and computing \shortname from the input graph is time-consuming. 
Thus, the online computation method cannot support a large number of \shortname queries. 

In this paper, we resort to index-based approaches. 
A straightforward solution is to compute all possible \shortnames and build a total index $\izero$ based on them. Instead of computing all \shortnames from the input graph, we take advantage of the nested property of the \shortname, which means that if $\alpha \geq \alpha^*$, $\beta \geq \beta^*$ and $\tau \geq \tau^*$, \shortname is a subgraph of $(\alpha^*,\beta^*)_{\tau^*}$-core. 
Specifically, for all possible $\alpha$ and $\beta$, we first find $(\alpha,\beta)_{1}$-core and then compute \shortname while gradually increasing strength level $\tau$. 
In this manner, we can compute all \shortnames and construct the index $\izero$. 
Although $\izero$ supports optimal retrieval of the vertex set of any \shortname, it still suffers from long construction time on large graphs. 
To devise more practical index-based approaches, we face the following challenges. 
\begin{enumerate}
    \item 
    When building index $\izero$, it is time-consuming to enumerate all butterflies containing the deleted edges.
    Also, the $\izero$ index construction algorithm is prone to visit the same \shortname subgraph repeatedly as it can correspond to different combinations of $\alpha, \beta$, and $\tau$.
    It is a challenge to speed up butterfly enumeration and avoid repeatedly visiting the same subgraphs during the construction of the total index $\izero$.
    \item 
    Due to the flexibility of the \shortname model, there are a large number of \shortnames corresponding to different combinations of $\alpha$, $\beta$, and $\tau$.
    The time cost of indexing all \shortnames becomes not affordable on large graphs. 
    It is also a challenge to strike a balance between building space-efficient indexes and supporting efficient and scalable query processing.
\end{enumerate}

\noindent
{\bf Our approaches.}
To address the first challenge, we extend the butterfly enumeration techniques in \cite{wang2020efficient} and propose novel computation sharing optimizations to speed up the index construction process of $\izero$. 
Specifically, we build a \texttt{Bloom}-\texttt{Edge}-\texttt{Index} (hereafter denoted by \texttt{BE-Index}) proposed in \cite{wang2020efficient} to quickly fetch the butterflies containing an edge.
The \texttt{BE-Index} captures the relationships between edges and $(2\times k)$-bicliques (also called \textit{blooms}).
When an edge is deleted, we can quickly locate the blooms containing this edge in the \texttt{BE-Index} and update the support of affected edges in these blooms accordingly. 
In addition, computation-sharing optimization is based on the fact that the same \shortname subgraph corresponds to various parameter combinations. If we realize the vertices in a subgraph have already been recorded, we can choose to skip the current parameter combination. 

To address the second challenge, we introduce space-efficient \gratings including $\ione$, $\itwo$, and $\ithree$, and train a feed-forward neural network to predict the most promising index to handle an \shortname query.
Instead of indexing all \shortnames, the \gratings $\ione$, $\itwo$, and $\ithree$ store the vertex sets of all \abcore, \IBG, and \AIG respectively. These \gratings are much smaller in size and require significantly less build time,
each of which can be used to handle \shortname queries. 
For example, to compute \shortname using $\itwo$, we fetch the vertices in \IBG and recover the edges of \IBG. Then, we iteratively remove the vertices not having enough engagement from \IBG until we find \shortname. 
However, the query processing performance based on each \grating is highly sensitive to parameters $\alpha$, $\beta$, and $\tau$. 
This is because the \gratings only store the vertices in \abcore, \IBG, and \AIG and the size difference between \shortname and each of these subgraphs is uncertain. 
We also observe that there are no simple rules to partition the parameter space so that queries from each partition can be efficiently handled by one type of index. 
This motivates us to resort to machine learning techniques and train a feed-forward neural network as the classifier to predict the optimal choice of the index for each incoming query of \shortname. 
Since we aim to minimize the query time instead of accuracy, we propose a scoring function, \textit{time-sensitive-error}, to tune the hyper-parameters of the classifier. 
The experiment results show that the resulting hybrid computation algorithm significantly outperforms the query processing algorithms based on $\ione$,$\itwo$, and $\ithree$, and it is less sensitive to varying parameters. 

\noindent
{\bf Contribution.}
Our major contributions are summarized here: \\
$\bullet$ We propose the first cohesive subgraph model \longname on bipartite graphs which considers both tie strength and vertex engagement. The flexibility of our model allows it to capture unique and useful structures on bipartite graphs. \\
$\bullet$ We construct index $\izero$ to support optimal retrieval of the vertex set of any \shortname.
We also devise computation sharing and \beindex based optimizations to effectively reduce its construction time.\\
$\bullet$ 
We build \gratings that are more space-efficient and require significantly less build time. 
Also, we propose a learning-based hybrid computation paradigm to predict which index to choose to minimize the response time for an incoming \shortname query.
\\
$\bullet$ We validate the efficiency of proposed algorithms and the effectiveness of our model through extensive experiments on real-world datasets.
Results show that the \gratings are scalable and the hybrid computation algorithm on a well trained neural network can outperform the algorithms based on each \grating alone. 
\\

\noindent
{\bf Organization.}
The rest of the paper is organized as follows.
Section $2$ reviews the related work. 
Section $3$ summarizes important notations and definitions and
introduces \abcore and \longname. 
Section $4$ presents the online computation algorithm. 
Section $5$ and $6$ presents the total index $\izero$ and optimizations of the index construction process.
Section $7$ presents the learning-based hybrid computation paradigm. 
Section $8$ shows the experimental results and Section $9$ concludes the paper.  
\section{Related work}
In the literature, there are many recent studies on cohesive subgraph models on both unipartite graphs and bipartite graphs. 

\noindent
\textit{Unipartite graphs.}
\kcore \cite{seidman1983network,cheng2011efficient,khaouid2015k,zhang2018finding} and \ktruss \cite{cohen2008trusses,huang2014querying,shao2014efficient} are two of the most well-known cohesive subgraph models on general, unipartite graphs. 
Given a unipartite graph, \kcore is the maximal subgraph such that each vertex in the subgraph has at least $k$ neighbors.
\kcore models vertex engagement as degrees and assumes the importance of each tie to be equal.
However, on real networks, ties (edges) have different strengths and are not of equal importance \cite{granovetter1977strength}. 
As triangles are considered as the smallest cohesive units, the number of triangles containing an edge is used to model tie strength on unipartite graphs.
Thus, \ktruss is proposed to better model tie strength, which is the maximal subgraph such that each edge in the subgraph is contained in at least $(k-2)$ triangles.
The issue with \ktruss is that it does not tolerate the existence of weak ties, which is inflexible for modeling real networks. %
To consider both vertex engagement and tie strength, the ($k$,$s$)-core model is proposed in \cite{zhang2018discovering}. 
In addition, recent works studied problems related to variants of \kcore such as  radius-bounded \kcore on geo-social networks \cite{wang2018efficient}, core maintenance on dynamic graphs \cite{zhang2017fast},
core decomposition on uncertain graphs \cite{bonchi2014core,peng2018efficient}, 
and anchored \kcore problem \cite{bhawalkar2015preventing,zhang2017olak}. 
Variants of \ktruss are also studied including \ktruss communities on dynamic graphs \cite{huang2014querying},
\ktruss decomposition on uncertain graphs \cite{zou2017truss}, and anchored \ktruss problem \cite{zhang2018efficiently}.
However, these algorithms do not apply to bipartite graphs. Attempts to project the bipartite graph to general graphs will incur information loss and size inflation \cite{sariyuce2018peeling}. 

\noindent
\textit{Bipartite graphs.}
In correspondence to \kcore and \ktruss, \abcore \cite{ding2017efficient,liu2020efficient} and \kbitruss \cite{zou2016bitruss,wang2020efficient} are proposed on bipartite graphs.
\abcore is the maximal subgraph such that each vertex on the upper or lower level in the subgraph has at least $\alpha$ or $\beta$ neighbors. 
Just like \kcore on unipartite graphs, \abcore cannot distinguish weak ties from strong ties. 
On bipartite graphs, tie strength is often modeled as the number of butterflies (i.e., ($2 \times 2$)-bicliques) containing an edge because butterflies are viewed as analogs of triangles \cite{sanei2018butterfly,wang2014rectangle,wang2019vertex,wang2020efficient}. 
\kbitruss is the maximal subgraph such that each edge in the subgraph is contained in at least $k$ butterflies, which can model tie strength.
\kbitruss suffers from the same issue as its counterpart \ktruss does: it forcefully deletes all weak ties even if the incident vertices are strongly-engaged.
Other works for bipartite graph analysis using cohesive structures such as ($p$,$q$)-core \cite{DBLP:journals/corr/CerinsekB15}, fractional \kcore \cite{giatsidis2011evaluating} cannot be used to address these issues. In contrast to the above studies, we propose the first cohesive subgraph model \longname that considers both vertex engagement and tie strength on bipartite graphs.
\section{Problem Definition}
\begin{table}[htb]
\centering
\caption{Summary of Notations}
\scalebox{1.0}{
\begin{tabular}{c|c}
\noalign{\hrule height 1pt}
Notation & Definition \\ 
\noalign{\hrule height 0.6pt}
$G$ & a bipartite graph \\
$\alpha,\beta$ & the engagement constraints \\
$\tau$ & the strength level \\
$\nb(u,G)$ & the set of adjacent vertices of $u$ in $G$ \\
$deg(u,G)$ & the number of adjacent vertices of $u$ in $G$ \\
$ \BTF_{G} $   &   the number of butterflies in $G$  \\
$\sup(e)$ & the number of butterflies containing $e$ \\
$\engage(u)$ & the number of strong ties adjacent to $u$ \\
\shortname & the $\tau$-strengthened $(\alpha,\beta)$-core \\
$\izero$ & the decomposition-based index \\
$\ione$,$\itwo$,$\ithree$ & the \gratings \\
\noalign{\hrule height 1pt}
\end{tabular}
}
\vspace{-2mm}
\label{tab:notation}
\end{table}

In this section, we formally define our cohesive subgraph model \longname.  We consider an unweighted, undirected bipartite graph $G(V, E)$. $V(G)$ = $U(G) \cup L(G)$ denotes the set of vertices in $G$ where $U(G)$ and $L(G)$ represent the upper and lower layer, respectively. $E(G) \subseteq U(G) \times L(G)$ denotes the set of edges in $G$. We use $n$ = $|V(G)|$ to denote the number of vertices and $m$ = $|E(G)|$ to denote the number of edges. 
The maximum degree in the upper and lower layer is denoted as $\dmaxu$ and $\dmaxv$ respectively.
The set of neighbors of a vertex $u$ in $G$ is denoted as $\nb(u,G)$. The degree of a vertex is $deg(u,G)$ = $|\nb(u,G)|$. When the context is clear, we omit the input graph $G$ in notations. 

\begin{definition}
\label{def:abcore}
{\bf $(\alpha,\beta)$-core.} Given a bipartite graph G and degree constraints $\alpha$ and $\beta$, 
a subgraph $G'$ is the $(\alpha,\beta)$-core, denoted by $C_{\alpha,\beta}(G)$,
if ($1$) all vertices in $G'$ satisfy degree constraints,
i.e. $deg(u,G')\geq \alpha$ for each $u \in U(G')$ and 
$deg(v,G')\geq \beta$ for each $v \in L(G')$;
and ($2$) $G'$ is maximal, i.e. any subgraph $G'' \supseteq G'$ is not an $(\alpha,\beta)$-core. 
\end{definition}

\begin{definition}
\label{def:btf}
{\bf Butterfly.} In a bipartite graph G, given vertices $u,w \in U(G)$ and $v,x \in L(G)$, a butterfly $\BTF$ is the complete subgraph induced by $u,v,w,x$, which means both $u$ and $w$ are connected to $v$ and $x$ by edges. 
The total number of butterflies in G is denoted as $\BTF_G$.
\end{definition}
\noindent
$(\alpha,\beta)$-core is a vertex-induced subgraph model, which assumes that the edges are of equal importance. To better model the strength of an edge $e$, we define the support $\support (e)$ to be the number of butterflies containing $e$. 
\begin{definition}
\label{def:st}
{\bf Strong Tie.} Given an integer $\tau$, an edge $e \in E(G)$ is called a strong tie if $\support(e) \geq \tau$, where $\tau$ is called the strength level. Weak ties are the edges $e$ such that $\support (e) < \tau$.
\end{definition}
\begin{definition}
\label{def:eng}
{\bf Vertex Engagement.} Given a strength level $\tau$ and $u \in V(G)$, the engagement $\engage (u)$ is the number of strong ties incident to $u$. At strength level 0, $\engage(u)=deg(u,G)$.
\end{definition}

\noindent
If the engagement of an upper or lower vertex is at least $\alpha$ or $\beta$, we call it a {\bf strongly-engaged} vertex. 
Otherwise, it is a {\bf weakly-engaged} vertex. 
\begin{definition}
\label{def:abgcore}
{\bf \longname.} 
Given a bipartite graph G and engagement constraints $\alpha$ and $\beta$, and strength level $\tau$, 
a subgraph $G'$ is the \longname, denoted by \shortname, 
if ($1$) $\engage(u) \geq \alpha$ for each $u \in U(G')$ and 
$\engage(v) \geq \beta$ for each $v \in L(G')$;
and ($2$) $G'$ is maximal, i.e. any subgraph $G'' \supseteq G'$ is not a \longname.
\end{definition}

\noindent
\textbf{Problem Statement.}  Given a bipartite graph $G$ and parameters $\alpha,\beta$ and $\tau$, we study the problem of scalable and efficient computation of \shortname in $G$.
\begin{algorithm}[h!]
	\caption{OnlineComputation}
	\label{algo:compute_naive}
	\LinesNumbered
	\KwIn{$G,\alpha,\beta,\tau$} 
	\KwOut{\shortname}
	Compute $\support (e)$ foreach $e \in E(G)$  \\
    Compute $\engage (u)$ foreach $u \in V(G)$ \\
    {\em Peeling($G,\alpha,\beta,\tau, \support,\engage$)}\\
\textbf{return} $G$
\end{algorithm}
\begin{algorithm}[h!]
	\caption{\compute}
	\label{algo:peel} 
	\LinesNumbered
	\KwIn{$G,\alpha,\beta,\tau, \support,\engage$} 
	\KwOut{\shortname}
    \While{exists $u \in V(G)$ without enough engagement}{
        \ForEach{$ v \in \nb(u)$ }{
           \If{{$\support ((u,v)) \geq \tau$}}{
                $\engage (v) \gets \engage (v)-1$
           }
           \ForEach{$\BTF$ containing $(u,v)$}{
                \ForEach{edge $e' = (u',v') \in\BTF$ s.t. $e'\neq e$ and $\support (e') \geq \tau$}{
                $\support (e') \gets \support (e')-1$ \\
                \If{$\support (e') = \tau-1$}{
                    decrease $\engage (u')$ and $\engage (v')$ by 1\\
                }
                }
           }
           remove $(u, v)$ from $G$ \\
        }
    remove $u$ from $G$\\
    }
\textbf{return} $G$
\end{algorithm}
\section{The Online Computation Algorithm}
Given engagement constraints $\alpha$, $\beta$ and strength level $\tau$, the online algorithm to compute the \shortname is outlined in Algorithm  \ref{algo:compute_naive}. 
First, we compute the support of each edge $e$ using the algorithm in \cite{wang2019vertex} and count how many strong ties each vertex $u$ has.  
Then, Algorithm \ref{algo:peel} is invoked to iteratively remove the vertices without enough engagement along with their incident edges.  
The vertices in $U(G)$ and $L(G)$ are sorted by engagement and the edges are sorted by support. In this manner, we can always delete the vertices with the smallest engagement first and quickly identify which edges are strong ties. 
When an edge $e$ is removed due to lack of support (i.e., $sup(e)<\tau$), we go through all the butterflies containing $e$ and update the supports of the edges in these butterflies (lines 5-9).
Specifically, we do not need to update the support of the edges connected to weakly-engaged vertices, because they will be removed (line $10$). 
Neither do we update the support of weak ties because they do not contribute to any vertex engagement. 
In other words, we only update the support of strong ties between strongly-engaged vertices. 
When a strong tie becomes a weak tie, we decrease the engagement of its incident vertices (lines 4,9). 
The order of vertices and edges are maintained in linear heaps \cite{chang2019cohesive} after their engagement and support are updated.  
Here we evaluate the time and space complexity of Algorithm \ref{algo:compute_naive}. 
\begin{lemma}
The time complexity of Algorithm \ref{algo:compute_naive} is $O(\sum_{(u,v)\in E(G)}\sum_{w\in \nb(v)} min(deg(u),deg(w)))$ and the space complexity is $O(m)$. 
\end{lemma}
\begin{proof}
The butterfly counting process takes $O(\sum_{(u,v)\in E(G)} min(deg(u),deg(v)))$ time \cite{wang2019vertex}. 
After the support of each edge is computed, it takes $O(m)$ time to compute the engagement for each vertex. 
Then, we need to run Algorithm \ref{algo:peel}. For each weakly engaged vertex $u$, we need to delete all its incident edges, which dominates the time cost of Algorithm \ref{algo:peel}. 

For each deleted edge $(u,v)$, we need to enumerate the butterflies containing it. 
Let $w$ be a vertex in $\nb(v,G) \setminus \{u\} $. 
For each vertex $x$ in $\nb(u,G) \cap \nb(w,G)$, the induced subgraph of $\{u,v,w,x\}$ is a butterfly. The set intersection (computing $\nb(u,G)\cap \nb(w,G)$) can be implemented in $O(min(deg(u),deg(w))$ time by using a $O(m)$ hash table to store the neighbor set of each vertex.

Thus, the butterfly enumeration for each delete edge $(u,v)$ takes $O(\sum_{w\in \nb(v)} min(deg(u),deg(w)))$ time.
As each edge can only be deleted once, the total time complexity of the butterfly enumeration for all deleted edges takes $O(\sum_{(u,v)\in E(G)}\sum_{w\in \nb(v)} min(deg(u),deg(w)))$ time. Thus, the time complexity of Algorithm \ref{algo:compute_naive} is $O(\sum_{(u,v)\in E(G)}\sum_{w\in \nb(v)} min(deg(u),deg(w)))$ which is denoted as $T_{peel}(G)$ hereafter.

We store the neighbors of each vertex as adjacency lists as well as the support of edges and engagement of vertices, which in total takes $O(m)$ space. Therefore, the space complexity of Algorithm \ref{algo:compute_naive} is $O(m)$.  
\end{proof}

\begin{figure}[htb]
\centering  
\subfigure[$\izero$ based on Figure \ref{fig:moltivation}]{
\label{fig.index}
\includegraphics[width=0.43\textwidth]{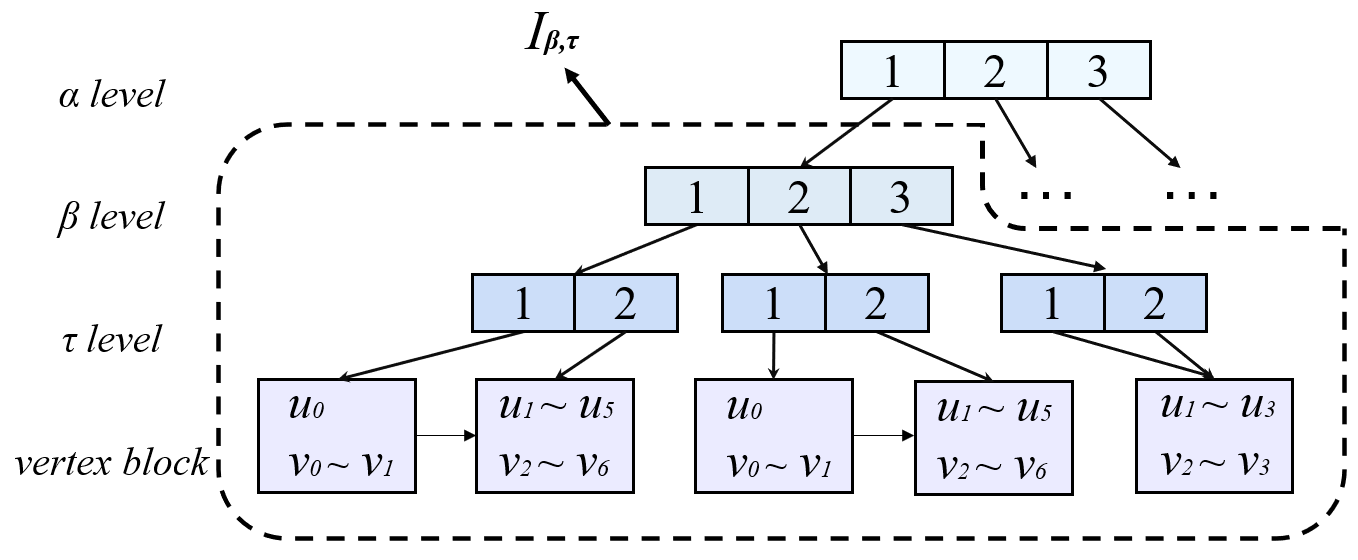}}
\subfigure[the relationships of indexes]{
\label{fig.cube}
\includegraphics[width=0.26\textwidth]{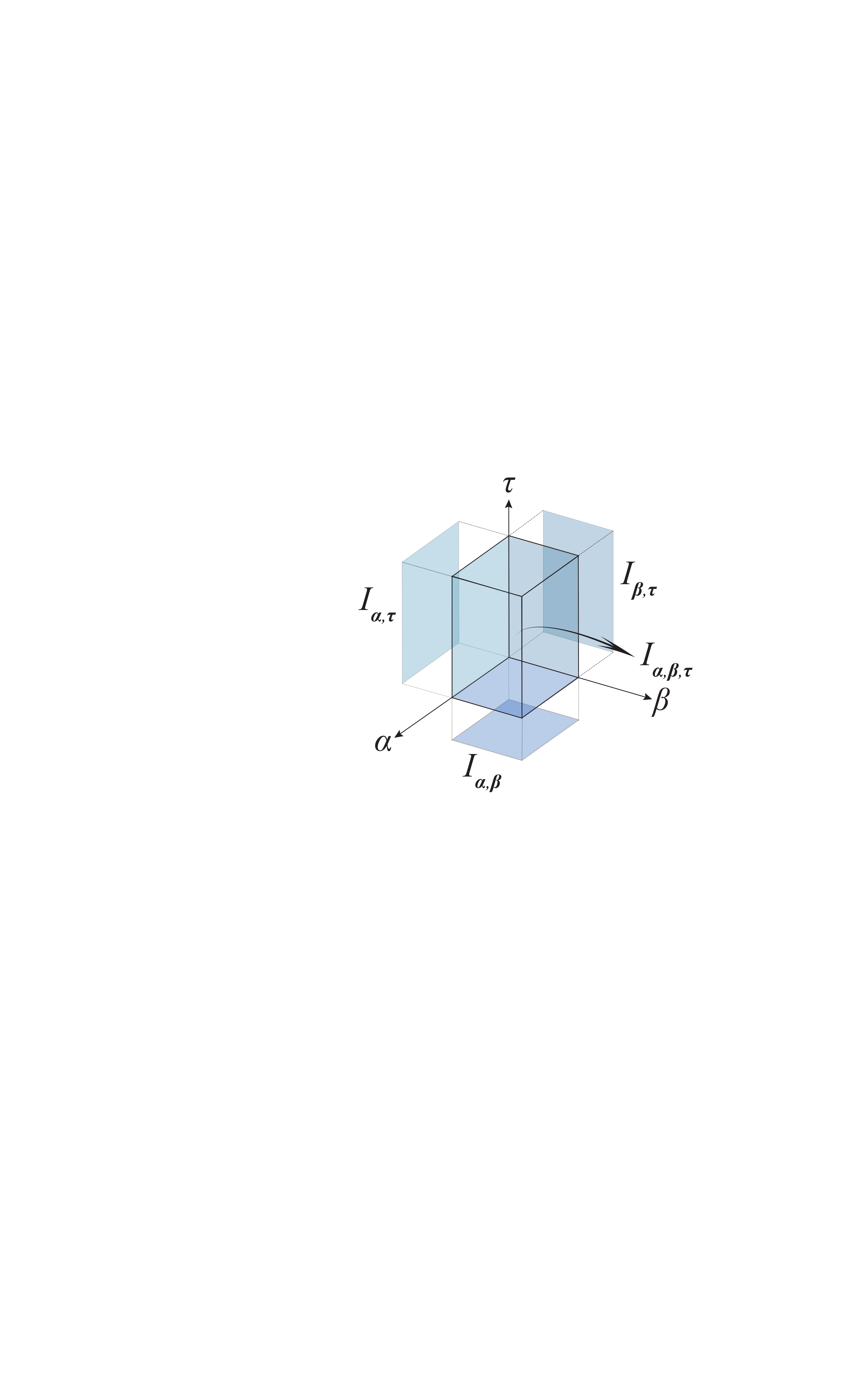}}
\subfigure[Computation sharing]{ 
\label{fig.share}
\includegraphics[width=0.26\textwidth]{compsharing.pdf}}
\caption{Illustrating our ideas}
\label{Fig.queries}
\end{figure}
\section{The Decomposition Based Total Index}
Given $\alpha$, $\beta$, and $\tau$, Algorithm \ref{algo:compute_naive} computes the \shortname from the input graph, which is slow and cannot handle a large number of queries. 
In this section, we present a decomposition algorithm that retrieves all \shortnames and build a total index based on the decomposition output to support efficient query processing.

\begin{algorithm}[t!]
	\caption{Decomposition}
	\label{algo:naivedecomp}
	\LinesNumbered
	\KwIn{$G(V=(U,L),E)$}
	\KwOut{$\tau_{max}(\alpha,\beta,u)$, for all $\alpha,\beta$, $\forall u \in V(G)$}
	$\alpha \gets 1$, $\beta \gets 1$, $\tau \gets 1$ \\
	Compute  $\support (e)$, $\forall$ $e \in E(G)$  \\
    Compute  $\engage (u)$, $\forall$ $u \in V(G)$ \\
    \While{$(\alpha,1)_{1}$-core in $G$ is not empty}{
        $\compute((\alpha,1)_{1}\textrm{-}core, \alpha, 1, 1, \support, \engage);\beta\gets1$ \\ 
        \While{$(\alpha,\beta)_{1}$-core in $G$ is not empty}{
            $\support' \gets \support$;\ $\engage' \gets \engage$;\ $\tau\gets1$ \\
            $\compute((\alpha,\beta)_{1}\textrm{-}core, \alpha, \beta, 1, \support ',  \engage ')$\\ 
            \While{ $(\alpha,\beta)_{\tau}$-core in $G$ is not empty}{ 
                $\support'' \gets \support'$;\ $\engage'' \gets \engage'$ \\
                $\compute((\alpha,\beta)_{\tau}\textrm{-}core, \alpha, \beta, \tau, \support '', \engage '')$, \ \ add $\tau_{max}(\alpha,\beta,u) \gets \tau-1 $ before line 2\\
                $\tau \gets \tau+1$ \\
                \ForEach{$e$=$(u',v') \in E(G)$, $\support (e)$=$\tau$}{
                    $\engage ''(u')\gets \engage ''(u')-1$ \\
                    $\engage ''(v')\gets \engage ''(v')-1$
                }
            }
            $\beta \gets \beta+1$
        }
        $\alpha \gets \alpha+1$
    }
\textbf{return} $\tau_{max}(\alpha,\beta,u)$, for all $\alpha,\beta$, $\forall u \in V(G)$ 
\end{algorithm}


\vspace{2mm}
\noindent
{\bf The decomposition algorithm.} 
The following lemma is immediate based on Definition \ref{def:abgcore}, which depicts the nested relationships among \shortnames. 
\begin{lemma}
\shortname $\subseteq  (\alpha',\beta')_{\tau'}$-core
if $ \alpha \geq \alpha' $, $ \beta \geq \beta' $, and $ \tau \geq \tau' $.
\label{lemma:nest}
\end{lemma}
Based on Lemma \ref{lemma:nest}, if a vertex $u$ is in $(\alpha, \beta)_{\tau'}$-core, $u$ is also in \shortname if $\tau < \tau'$. Thus, in the decomposition, for given vertex $u$ and $\alpha$, $\beta$ values, we aim to retrieve the maximum $\tau$ value such that $u$ is in the corresponding $(\alpha, \beta)_\tau$-core, namely, $\tau_{max}(\alpha,\beta,u) \textnormal{=} \max\{\tau \vert  u \in (\alpha, \beta)_\tau\textrm{-}core \} $. 
For each vertex $u$ and all possible combinations of $\alpha$ and $\beta$, it is only necessary to store $u$ in $(\alpha,\beta)_{\tau'}$-core to build a space compact index, where $\tau'$ = $\tau_{max}(\alpha,\beta,u)$ and $u \in$ $(\alpha, \beta)_{\tau}$-core can be implied if $\tau < \tau'$.
Algorithm \ref{algo:naivedecomp} is devised for \shortname decomposition which applies three nested loops to go through all possible $\alpha, \beta, \tau$ combinations. 
Note that when computing the \shortname, we record $\tau_{max}(\alpha,\beta,u)$ for each vertex $u$. Specifically, for a vertex $u$ contained in $(\alpha,\beta)_{\tau_0}$-core but is removed when computing $(\alpha,\beta)_{\tau_0+1}$-core, we assign $\tau_0$ to $\tau_{max}(\alpha,\beta,u)$ (line 11).
Here we evaluate the time and space complexity of Algorithm \ref{algo:naivedecomp}. 
\begin{lemma}
The time complexity of Algorithm \ref{algo:naivedecomp} is 
$O(\sum_{\alpha =1 }^{\alpha_{max}} \sum_{\beta=1}^{\beta_{max}(\alpha)} T_{peel}(\abi))$,
where $\alpha_{max}$ is the maximal $\alpha$ such that $(\alpha,1)_1$-core exists and $\beta_{max}(\alpha)$ is the maximal $\beta$ such that \ABI exists.
Algorithm \ref{algo:naivedecomp} takes up $O(m)$ space.
\end{lemma}

\begin{proof}
In the outer while-loop, we start from the input graph $G$ and compute \AII, which takes $\peelG$.
In the middle while-loop, we calculate \ABI from \AII for all possible $\alpha$, which takes  $ \sum_{\alpha =1 }^{\alpha_{max}} T_{peel}(\aii)$. 
The dominant cost occurs in the innermost while-loop when we run the Peeling algorithm on \ABI for all possible $\alpha$ and $\beta$. 
As iterative removing edges and vertices in \ABI until it is empty takes $T_{peel}(\abi)$, 
the overall time complexity is $O(\sum_{\alpha =1 }^{\alpha_{max}} \sum_{\beta=1}^{\beta_{max}(\alpha)} T_{peel}(\abi))$.

At any time during the execution of this algorithm, we always store $G$, \AII, \ABI, and \shortname in memory, which takes $O(m)$ space. 
We also store the support of edges and engagement of vertices in these graphs, which also takes $O(m)$ space. 
Thus, the total space complexity is $O(m)$. 
\end{proof}

\noindent
{\bf Decomposition-based index.} Based on the decomposition results, a four level index $\izero$ can be constructed to support query processing as shown in Figure \ref{fig.index}.
\begin{enumerate}
    \item[$\bullet$] \textit{$\alpha$ level.} The $\alpha$ level of $\izero$ is an array of pointers, each of which points to an array in the $\beta$ level. The length of the array in the $\alpha$ level is $\alpha_{max}$. The $k_{th}$ element is denoted as $\izero[k]$. 
    \item[$\bullet$] \textit{$\beta$ level.} 
    The $\beta$ level has $\alpha_{max}$ arrays of pointers. The array pointed by $\izero[k]$ has length $\beta_{max}(\alpha)$. The $j_{th}$ pointer in the $k_{th}$ array is denoted as $\izero[k][j]$, which points to an array in the $\tau$ level. 
    \item[$\bullet$] \textit{$\tau$ level.}
     The $\tau$ level has $\sum_{i=1}^{\alpha_{max}}\beta_{max}(i) $ arrays of pointers to vertex blocks, corresponding to all pairs of $\alpha,\beta$. The array pointed by $\izero[k][j]$ has length $\tau_{max}(\alpha,\beta) = \max\{\tau \vert \abg \ exists \} $. 
     \item[$\bullet$] \textit{vertex blocks.} The fourth level of $\izero$ is a singly linked list of vertex blocks. Each vertex block corresponds to a set of $\alpha,\beta,\tau$ value, which contains all vertex $u$ such that $\tau_{max}(\alpha,\beta,u) = \tau$ along with a pointer to the next vertex block. 
     The vertex blocks with the same $\alpha$ and $\beta$ are sorted by the associated $\tau$ values and each of them has a pointer to the next.  
     Among them, the vertex block with the largest $\tau$ has its pointer pointing to null. 

\end{enumerate}

We can construct index $\izero$ from the output of Algorithm \ref{algo:decomp_query}. 
Given $\alpha$ and $\beta$, we store all vertices $u$ in the same vertex block if they have the same $\tau_{max}(\alpha,\beta,u)$, so each vertex block has an associated $(\alpha,\beta,\tau)$ value. 
In each $\izero[\alpha][\beta][\tau]$, we store the address of the vertex block associated with $(\alpha,\beta,\tau')$, where $\tau'$ is the smallest integer such that $\tau \leq \tau'$.
The index construction time is linear to the size of decomposition results, which is bounded by the time complexity of Algorithm \ref{algo:decomp_query}. 

\begin{algorithm}[t]
	\caption{\Query}
	\label{algo:decomp_query}
	\LinesNumbered
	\KwIn{$\izero,\alpha,\beta,\tau,G$}
	\KwOut{\shortname}
	$U',V',E' \gets \emptyset$ \\
	\If{$\izero.size <\alpha$ or $\izero[\alpha].size <\beta$ or $\izero[\alpha][\beta].size < \tau$ }{
        return $\emptyset$ \\
	}
	$ ptr\gets \izero[\alpha][\beta][\tau]$  \\
    \While{$ptr$ is not null }{
        $v \gets $ vertices in vertex block pointed by $ptr$ \\ 
        \If{$v \in U(G)$}{
            $U' \gets U' \cup v$
        }\Else{
            $V' \gets V' \cup v$
        }
        $ptr \gets $ the address of the next vertex block 
    } 
    $E' \gets E(G) \cap (U' \times V')$ \\ 
\textbf{return} $G'=(U',V',E')$
\end{algorithm}
\begin{lemma}
\label{lemma:izero}
The space complexity of index $\izero$ is 
$O(\sum_{i=1}^{\alpha_{max}} \sum_{j=1}^{\beta_{max}(i)} (\tau_{max}(i,j)+n))$,
where  $\tau_{max}(i,j)$ be the maximal $\tau$ in all $(i,j)_{\tau}$-cores. 
\end{lemma}

\begin{proof}
By construction, the space complexity of the first two levels of pointers is bounded by that of $\tau$ level. 
The space complexity of $\tau$ level is $\sum_{i=1}^{\alpha_{max}} \sum_{j=1}^{\beta_{max}(i)} \tau_{max}(i,j)$. 
Given vertex $u$, let $\alpha_{max}(u)$ be the maximal $\alpha$ such that $u \in \abg$.
The space complexity of the vertex blocks is  
$\sum_{u\in V(G)} \sum_{i=1}^{\alpha_{max}(u)}\max \{\beta | u \in (i,\beta)_1\textnormal{-}core\} \leq \sum_{i=1}^{\alpha_{max}} \sum_{j=1}^{\beta_{max}(i)} n$. Adding it to the space complexity of level $\tau$ completes the proof.
\end{proof} 

\noindent
{\bf Index based query processing.} 
When the index $\izero$ is built on a bipartite graph $G$, Algorithm \ref{algo:decomp_query} outlines how to restore \shortname given $\alpha,\beta$, $\tau$ and $\izero$. 
First, it checks the validity of the input parameters. 
If the queried \shortname does not exist, it terminates immediately (lines 2,3). 
Otherwise, it collects the vertices of \shortname from $\izero$ and restores the edges in the queried subgraph. 
\begin{lemma}
\label{lemma:query}
Given a graph $G$ and parameters $\alpha,\beta$ and $\tau$, 
Algorithm \ref{algo:decomp_query} retrieves 
$V(\abg)$ from index $\izero$ in $O(|V(\abg)|)$ time. 
The edges in \shortname can be retrieved in $O(\sum_{v \in V(\abg)}deg(v)$ time after obtaining the vertex set.
\end{lemma}
\begin{proof}
As each vertex block only stores the vertices with one given $\tau$ value, the vertex blocks pointed by $\izero[\alpha][\beta][\tau']$ where $\tau' \geq \tau$ give us all the vertices in \shortname, which takes $O(|V(\abg)|)$ time. 
To restore the edges in \shortname, for each vertex $v$ in \shortname, we go through each of its neighbors in $G$ and check if it is in $V(\abg)$. This takes $O(\sum_{v \in V(\abg)}(deg(v)))$ time.
\end{proof} 

\noindent
According to Lemma \ref{lemma:query}, given $\alpha,\beta$ and $\tau$, the vertex set of the \shortname is retrieved in optimal time. 
\begin{example}
Figure \ref{fig.index} illustrates the $I_{\alpha, \beta, \tau}$ index for the bipartite graph in Figure \ref{fig:moltivation}. When querying $(1,2)_{1}$-core, we start with the vertex block pointed by $\izero[1][2][1]$ ($u_0,v_0$ and $v_1$). 
Then we keep collecting the vertices until we have fetched the vertices pointed by $\izero[1][2][2]$ ($u_1$ to $u_5$ and $v_2$ to $v_6$) . All the collected vertices provides the final solution to vertex set of $(1,2)_{1}$-core, which are $u_0$ to $u_5$ and $v_0$ to $v_6$. 
\end{example} 
\section{Optimizations of Index Construction }
The above decomposition algorithm has these issues:
($1$) the same subgraph can be computed repeatedly for different $\alpha$ and $\beta$ values. 
For example, if $(1,1)_{\tau}$-core is the same subgraph as $(1,2)_{\tau}$-core, 
then we will compute it twice when $\beta$=$1$ and $\beta$=$2$. 
($2$) when removing an edge $e$, we need to enumerate all the butterflies containing $e$.
The basic implementation of butterfly enumeration is inefficient, which involves finding three connected vertices first and then check if a fourth vertex can form a butterfly with the existing ones.
We devise computation-sharing optimizations to address the first issue, and adopt the \texttt{Bloom}-\texttt{Edge}-\texttt{Index} proposed in \cite{wang2019vertex} to speed up the butterfly enumeration process.

\noindent
{\bf Computation sharing optimizations.}
In this part, we reduce the times of visiting the same \shortname subgraphs by skipping some combinations of $\alpha$, $\beta$, and $\tau$, while yielding the same decomposition results. 

\noindent
{$\bullet$ \it Skip computation for $\tau$.}
In Algorithm \ref{algo:naivedecomp}, if vertex $u$ is removed when computing $(\alpha,\beta)_{\tau+1}$-core from \shortname, we conclude that $\tau_{max}(\alpha,\beta,u)$=$\tau$. However, if both \shortname and $(\alpha,\beta)_{\tau+1}$-core are already visited for other $\beta$, this process is redundant and the current $\tau$ value can be skipped. 
Specifically, in the innermost while loop (lines $9$-$15$), we can use an array $\beta_{min}[\tau]$ to store the minimal lower engagement of \shortname for each $\tau$. 
If $\beta_{min}[\tau] \geq \beta$, then the current \shortname has already been visited.
We only compute $\tau_{max}(\alpha,\beta,u)$ values when one of \shortname and $(\alpha,\beta)_{\tau+1}$-core is not visited. Otherwise, we skip the current $\tau$ value. 
The following lemma explains how to correctly obtain multiple $\tau_{max}(\alpha,\beta,u)$ values when removing one vertex.
\begin{lemma}
\label{lemma:sharing}
Given $\alpha,\beta,\tau$ and graph $G$, let $u$ be a vertex in $(\alpha,\beta)_{\tau}$-core but not in $(\alpha,\beta)_{\tau+1}$-core. If there exists an integer $k$ such that $(\alpha,\beta)_{\tau}$-core = $(\alpha,\beta+k)_{\tau}$-core, then $u \not \in$ $(\alpha,\beta+k)_{\tau+1}$-core. 
\end{lemma}
This lemma is immediate from Lemma \ref{lemma:nest}, because if vertex $u$ is in $(\alpha,\beta+k)_{\tau+1}$-core, then it must also be contained in $(\alpha,\beta)_{\tau+1}$-core, which contradicts our assumption. 
Therefore, for vertices like $u$, we can conclude that $\tau_{max}(\alpha, \beta',u)$=$\tau$, for all $\beta \leq \beta' \leq \beta+k$. In this way, we fully preserve the decomposition outputs of Algorithm \ref{algo:naivedecomp}. 

\noindent
{$\bullet$ \it Skip computation for $\alpha$ and $\beta$.}
To skip some $\beta$ values, we keep track of the minimal engagement of lower level vertices of the visited \shortnames in the middle while-loop (lines $7$-$16$) in Algorithm \ref{algo:naivedecomp} as $\beta^*$. At line $16$, if $\beta^* > \beta$, then $\beta$ should be directly increased to $\beta^*+1$ (the first value which is not computed yet) and values from $\beta$ to $\beta^*$ are skipped.
This is because for all $ \beta \leq \beta' \leq \beta^*$, the decomposition process are exactly the same. 
Likewise, to skip some $\alpha$ values, we record the minimal engagement of upper-level vertices of the visited \shortnames in the outermost while-loop (lines $5$-$17$) of Algorithm \ref{algo:naivedecomp} as $\alpha^*$. At line $17$, if $\alpha^* > \alpha$, then $\alpha$ should be directly increased to $\alpha^*+1$ and values from $\alpha$ to $\alpha^*$ are skipped. 

\begin{example}
As shown in Figure \ref{fig.index},
when $\beta$=$1$, we remove $u_0,v_0$ and $v_1$ when computing $(1,1)_2$-core from $(1,1)_1$-core. 
The minimal engagement of lower vertices in $(1,1)_1$-core and $(1,1)_2$-core are $2$, so the array $\beta_{min}$ is $[2,2]$.
This means that $\tau_{max}(1,\beta',u)$ = $1$ and $\tau_{max}(1,\beta',u')$ = $2$, where $u \in \{u_0,v_0,v_1 \}$ and $u' \in \{u_1,u_2,u_3,u_4,u_5,v_2,v_3,v_4,v_5,v_6\}$ and $\beta' \in \{1,2\}$. 
When $\beta$=$2$, we infer that $(1,2)_1$-core and $(1,2)_2$-core are already visited based on array $\beta_{min}$, so we can skip $\beta$=$2$. 
When $\beta$=$3$, the current $\beta$ value exceeds the values in $\beta_{min}$, so it cannot be skipped. 
\end{example}

\noindent
{\bf Bloom-Edge-Index-based optimization.}
During edge deletions of Algorithm \ref{algo:naivedecomp}, we need to repeatedly retrieve the butterflies containing the deleted edges. 
To efficiently address this, We deploy a \texttt{Bloom}-\texttt{Edge}-\texttt{Index} (hereafter denoted as \beindex) proposed in \cite{wang2020efficient} to facilitate butterfly enumeration. 
Specifically, a bloom is a $2 \times k$-biclique, which contains $\frac{k \times (k\textnormal{-}1)}{2}$  butterflies. 
Each edge in the bloom is contained in $k\textnormal{-}1$ butterflies. 
The \beindex compresses butterflies into blooms and keeps track of the edges they contain. 
The space complexity of \beindex and the time complexity to build it are both $O( \sum_{e=(u,v) \in E(G)} min(deg(u),deg(v))$ \cite{wang2020efficient}. 
Hereafter we also use $\btftime$ to represent $\sum_{e=(u,v) \in E(G)} min(deg(u),deg(v))$.
Deleting an edge $e$ based on \beindex takes $O(sup(e))$ time, where $sup(e)$ is the number of butterflies containing $e$. 
This is because when $e$ is deleted, \beindex fetches the associated blooms and updates the support number of the affected edges in these blooms. In total, there are $O(sup(e))$ affected edges if edge $e$ is deleted. 
Here we evaluate the \beindex's impact on the overall time and space complexity of Algorithm \ref{algo:naivedecomp}. 
\begin{lemma}
\label{thm:bloom}
By adopting the \beindex for edge deletions, 
the time complexity of Algorithm \ref{algo:naivedecomp} becomes 
$O(\btftime)$+$O(\sum_{\alpha =1 }^{\alpha_{max}} \sum_{\beta=1}^{\beta_{max}(\alpha)} \BTF_{\abi}$), 
where $\BTF_{\abi}$ is the number of butterflies in \ABI.
\end{lemma}
\begin{proof}
In the innermost loop of Algorithm \ref{algo:decomp_query}, we remove the edges from \ABI to get \shortname. 
As each edge deletion operation takes $sup(e)$ \cite{wang2020efficient}, it takes $\sum_{e \in E(\abi)}$ = $O( \BTF_{\abi} )$ to compute \shortname from \ABI. 
As we are doing this for all possible $\alpha$ and $\beta$, the time complexity within the while-loops becomes $O(\sum_{\alpha =1 }^{\alpha_{max}} \sum_{\beta=1}^{\beta_{max}(\alpha)} \BTF_{\abi}$). 
Adding the \beindex construction time to it completes the proof. 
\end{proof} 
\section{A Learning-based Hybrid Computation Paradigm}
Although the index $\izero$ supports the optimal retrieval of the vertices in the queried \shortname, it does not scale well to large graphs due to its long build time and large space complexity even with the related optimizations.
For instance, on datasets \texttt{Team}, \texttt{Wiki-en}, \texttt{Amazon}, and \texttt{DBLP},  the index $\izero$ cannot be built within two hours as evaluated in our experiments.
In this section, we present \gratings that selectively store the vertices of \shortname for some combinations of $\alpha,\beta$, and $\tau$.
We also train a feed-forward neural network on a small portion of the queries to predict the choice of \grating that minimizes the running time for each new incoming query. 

\begin{table*}[tbh]
\centering
\scalebox{0.85}{
\begin{tabular}{c|c|c|c} 
\noalign{\hrule height 1pt}
Index   & Space Complexity & Build Time  & Query Time  \\
\noalign{\hrule height 0.56pt}
$\izero$ & $O(\sum_{i=1}^{\alpha_{max}}\sum_{j=1}^{\beta_{max}(i)}(\tau_{max}(i,j)+n))$ & $O(\btftime\textnormal{+}\sum_{\alpha =1 }^{\alpha_{max}} \sum_{\beta=1}^{\beta_{max}(\alpha)}\BTF_{\abi})$ &  $O(\sum_{v \in V(\abg)}(deg(v)))$ \\
$\ione$ & $O(m)$ & $O(\delta \cdot m )$ & $O(T_{peel}(\ab))$     \\ 
$\itwo$ & $O( \sum_{j=1}^{\beta_{max}}(\tau_{max}(1,j)+n))$ & $O(\btftime\textnormal{+}\sum_{\beta=1}^{\beta_{max}} \BTF_{\ibi})$ & $O(T_{peel}(\ibg)$ \\
$\ithree$ & $O(\sum_{i=1}^{\alpha_{max}}(\tau_{max}(i,1)+n))$ &  $O(\btftime\textnormal{+}\sum_{\alpha =1 }^{\alpha_{max}} \BTF_{\aii})$ & $O(T_{peel}(\aig)$ \\ 
\noalign{\hrule height 1pt}
\end{tabular}
}
\caption{Space complexity, index construction time and query processing time of different indexes}
\label{tab:indexCompare}
\end{table*} 

\noindent
{\bf \gratings.} 
We introduce three \gratings $\ione$, $\itwo$, and $\ithree$ in this part. 
Each of them is a three-level index with two levels of pointers and one level of vertex blocks. The main structures of them are presented as follows. 


\noindent
$\bullet$ $\itwo$.
For all $\beta$, $\tau$, $\itwo[\beta][\tau]$ points to the vertices $u\in V(G)$ s.t. $\tau_{max}(1,\beta,u) = \tau$.
It is the component of $\izero$ where $\alpha$=$1$, which can fetch the vertices of all subgraphs of the form \IBG in optimal time. 

\noindent
$\bullet$ $\ithree$.
For all $\alpha$, $\tau$, $\ithree[\alpha][\tau]$ points to the  vertices $u\in V(G)$ s.t. $\tau_{max}(\alpha,1,u) = \tau$.
It is the component of $\izero$ where $\beta$=$1$, which can fetch the vertices of all subgraphs of the form \AIG in optimal time. 

\noindent
$\bullet$ $\ione$ consists of $\ione U$ and $\ione V$ to store the vertices in $U(G)$ and $L(G)$ separately.
$\ione U[\alpha][\beta]$ points to the vertices $u\in U(G)$ s.t. 
$\beta = \max\{\beta' | u \in (\alpha,\beta')\textnormal{-}core \} $ and $\ione V[\beta][\alpha]$ points to the vertices $v\in L(G)$ such that $\alpha = \max\{\alpha' | v \in (\alpha',\beta)\textnormal{-}core\} $. 
$\ione$ can fetch the vertices of any \abcore in optimal time. 

Note that, $\ione$ is proposed to support efficient \abcore computation as introduced in \cite{liu2020efficient} while $\itwo$ and $\ithree$ are essentially parts of $\izero$. For each type of \gratings, we analyze its construction time, space complexity, and the query time to compute \shortname based on it.
\begin{lemma}
The time complexity to build $\itwo$ is $O(\btftime$+$\sum_{\beta=1}^{\beta_{max}} \BTF_{\ibi})$
and the space complexity of $\itwo$ is $O( \sum_{j=1}^{\beta_{max}}(\tau_{max}(1,j)+n))$. 
It takes $T_{peel}(\ibg)$ time to compute \shortname using $\itwo$.
\end{lemma}
\begin{proof}
As discussed in Lemma \ref{thm:bloom}, the \beindex can significantly speed up butterfly enumeration during edge deletions, which takes $O(\btftime)$ to construct. 
Then, we fix $\alpha$ to one and run lines 6-16 of Algorithm \ref{algo:naivedecomp} to compute all \IBG. 
For each possible $\beta$, this process takes $O(\BTF_{\ibi})$ time, so in total it takes $O(\sum_{\beta=1}^{\beta_{max}} \BTF_{\ibi}))$ time. Adding it to the \beindex construction time ($O(\btftime)$) gives the time complexity of $\itwo$ construction. 

As $\itwo$ is the part of $\izero$ with $\alpha$ = $1$, its space is equal to the part of $\izero$ that is pointed by $\izero[1]$. The size of the arrays of pointers in $\itwo$ is bounded by $O(\sum_{j=1}^{\beta_{max}}(\tau_{max}(1,j))))$. The size of the vertex blocks is bounded by the number of vertices in all \IBG, which is $O(\sum_{j=1}^{\beta_{max}}(|V(\ibg)| )) = O(\sum_{j=1}^{\beta_{max}} n)$. Therefore, the space complexity of $\itwo$ is $O( \sum_{j=1}^{\beta_{max}}(\tau_{max}(1,j)+n))$. 
\end{proof}

\noindent
In order to query \shortname based on $\itwo$, we first find the \IBG from $\itwo$ and then compute \shortname from \IBG. 
\begin{lemma}
The query time of \shortname based on $\itwo$ is $T_{peel}(\ibg)$. 
\end{lemma}
\begin{proof}
Given engagement constraints $\alpha,\beta$ and strength level $\tau$, let $G'$ be the \IBG on bipartite graph $G$. 
First, it takes $O(|V(G')|$ time to fetch the vertices in \IBG from $\itwo$. 
Restoring the edges of \IBG from $G$ takes $O(\sum_{u\in V(G')} deg(u,G'))$ time. 
Then, we call the \compute algorithm on \IBG to compute \shortname, which takes $O(T_{peel}(\ibg))$ time.
\end{proof}
\begin{example}
In Figure \ref{fig.index}, the component of $\izero$ wrapped in dotted line is the $\itwo$ of the graph in Figure \ref{fig:moltivation}. If $(2,2)_2$-core is queried, we first to obtain $(1,2)_2$-core from $\itwo$ and compute $(2,2)_2$-core by calling the peeling algorithm.
\end{example}

Note that index $\ithree$ is symmetric to index $\itwo$, with $\alpha$ and $\beta$ switched.  
It is immediate that 
it takes $O(\btftime$+$\sum_{\alpha=1}^{\alpha_{max}} \BTF_{\aii})$ time to construct $\ithree$ and its space complexity is $O( \sum_{i=1}^{\alpha_{max}}(\tau_{max}(i,1)+n))$. 
It takes $T_{peel}(\aig)$ time to compute \shortname using $\ithree$.
As for $\ione$, it takes $O(\delta \cdot m)$ time to construct and its space complexity is $O(m)$, where $\delta$ is the degeneracy of the graph \cite{liu2020efficient}. 
To compute \shortname using $\ione$, we fetch the vertices of \abcore and restore the edges in $O(|V(G')|+O(\sum_{u\in V(G')} deg(u,G'))$ time ($G'$ = \abcore ). 
Then, we call the peeling algorithm on \abcore to compute \shortname, which takes $O(T_{peel}(\ab))$ time. 

The sizes of \gratings can be considered as the projections of $\izero$ onto 3 planes, as depicted in Figure \ref{fig.cube}. We also summarize the space complexity, build time, and query time of \gratings in Table \ref{tab:indexCompare}. 

\begin{figure}[htb]
\centering
\scalebox{0.75}{
\includegraphics[width=0.7\textwidth]{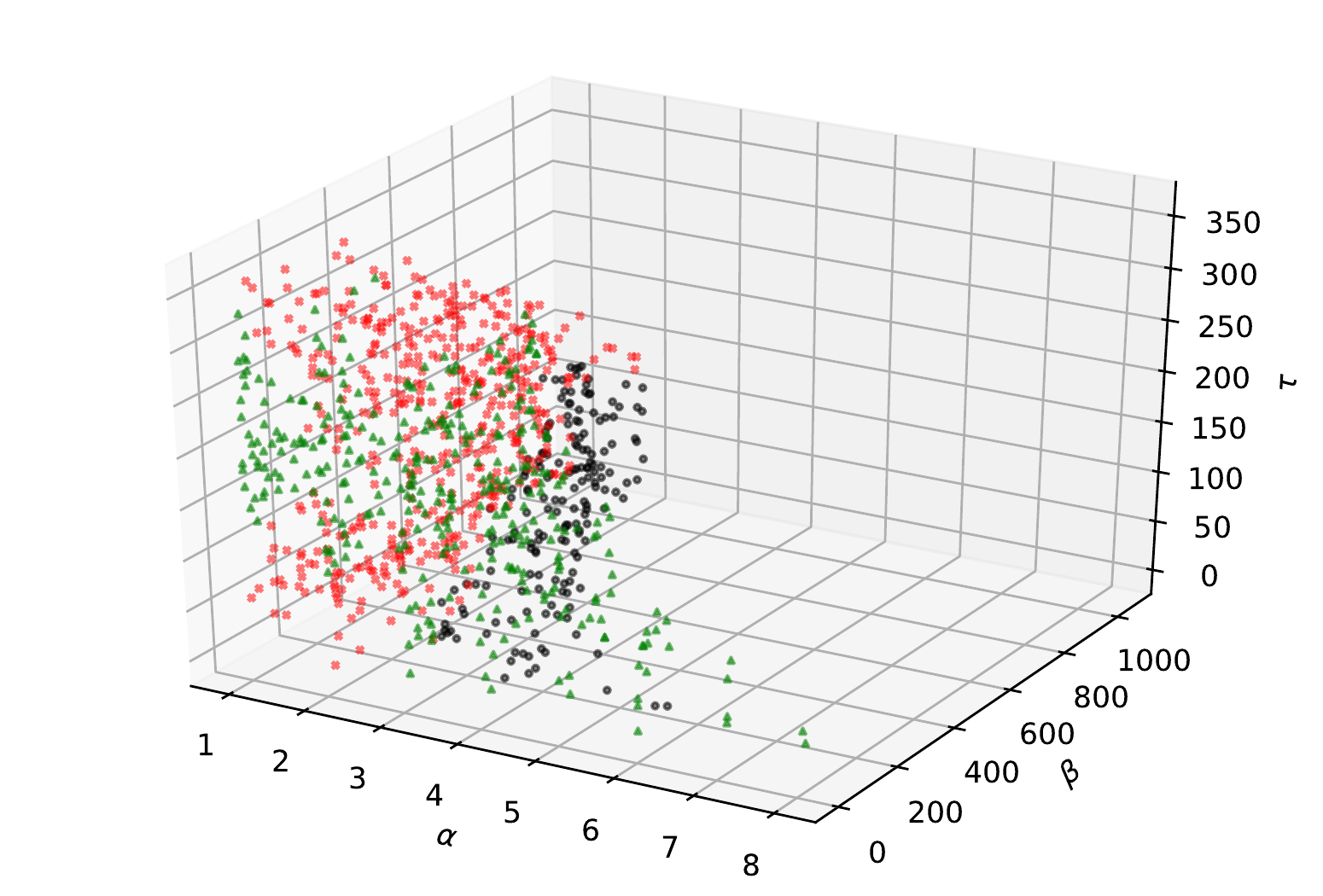}
}
\caption{Motivation example for learning-based query processing (\texttt{DBpedia-Team})}
\label{fig:learning}
\end{figure}

\noindent
{\bf Learning-based hybrid query processing.}
As $\ione,\itwo$, and $\ithree$ do not store all the decomposition results like $\izero$, the construction of these indexes is more time and space-efficient than $\izero$. 
However, the reduced index computation inevitably compromises the query processing performance.
This is because the \gratings only store the vertices in \abcore, \IBG, and \AIG. 
Clearly, computing the \shortname based on these \gratings results in different response time.
To better illustrate this point, we plot all parameter combinations on dataset \texttt{DBpedia-team} and give each combination a color based on which query processing algorithm performs the best in Figure \ref{fig:learning}.
For ease of presentation, we denote the query processing algorithms based on index $\ione$, $\itwo$, and $\ithree$ as $\qone$, $\qtwo$ and $\qthree$ respectively.
The red points indicate that $\qtwo$ is the fastest among the three. 
The green ones represent the win cases for $\qthree$ and the black points are the cases when $\qone$ is the best.
Evidently, the points of different colors are mingled together and distributed across the parameter space.
This suggests that finding simple rules to partition the parameter space is not promising in deciding which of ${\qone, \qtwo, \qthree}$ is the fastest. 
Hence, we formulate it as a classification problem and resort to machine learning techniques to solve this problem. 
\begin{algorithm}[thb]
    \LinesNumbered
    	\caption{Hybrid Computation Algorithm}
	\label{algo:hyrid}
	\tcp{\textbf{Offline training:}} 
	\setcounter{AlgoLine}{0} 
	\KwIn{ $G:$ Input bipartite graph} 
	\KwOut{Neural network $C:D \to \{\qone,\qtwo,\qthree\}$ } 
    Build $\itwo,\ithree$ and $\ione$ for $G$\\
    \ForEach{$q\in \{\textnormal{$N$ random queries}\} $ on $G$}{
    $feature(q)\gets[\alpha,\beta,\tau\textnormal{ of $q$}]$ \\ 
    $label(q)\gets$ the fastest algorithm in $\{\qone, \qtwo, \qthree \}$\\
    }
      $X=[feature(q)]$, $q\in$ $N$ queries run on $G$ \\
      $Y=[label(q)]$ , $q\in$ $N$ queries run on $G$ \\
      $C \gets$ trained neural network on $X,Y$ \\
	\tcp{\textbf{Online query processing:}} 
	\KwIn{Query parameters: $\alpha,\beta,\tau$}
	\KwOut{\shortname in $G$} 
	\setcounter{AlgoLine}{0}
	$Q_{pred}\gets$ $C.predict(\alpha,\beta,\tau)$ \\
    Run $Q_{pred}$ to compute \shortname \\
\textbf{return } \shortname  \\
\end{algorithm}

We introduce a hybrid computation algorithm (Algorithm \ref{algo:hyrid}, denoted by $Q_{hb}$), which selects from $\{ \qone,\qtwo,\qthree \}$ based on the query parameters $\alpha,\beta$, and $\tau$.
In the offline training phase, we build $\itwo,\ithree$ and $\ione$ on $G$ and obtain the runtime of $\qone,\qtwo$, $\qthree$ on $N$ queries, 
where $N$ is chosen to be less than $5\%$ of all possible queries. 
The label of a query is the algorithm that responds to it in the shortest time.
Then, we train a feed-forward neural network $C$ on the $N$ labeled query instances. 
In the online query processing phase, given a new query of \shortname, the trained neural network makes a prediction based on $\alpha,\beta \textnormal{, and }\tau$. Then, we use the predicted query processing algorithm to compute \shortname.

Here we detail how to train the feed-forward neural network.
We impose only one hidden layer in $C$ to avoid over-fitting. 
The important hyper-parameters of $C$ include the number of hidden units $H$ and the type of optimizer.
We use $5$-fold cross-validation to evaluate the above hyper-parameters. 
Specifically, we split the $N$ labeled queries into $5$ partitions and each time we take one partition as the validation set and the remainder as the training set. 
For each parameter setting, we build a classifier on the training sets for $5$ times and calculate a performance metric on the validation set. 
In our model, we define a \textit{time-sensitive error} on the validation set as the performance metric, which calculates a weighted mis-classification cost w.r.t the actual query time.
Let $i_k, j_k \in \{ 1,2,3 \}$ (encoding of $\qone,\qtwo,\qthree$) be the predicted class and the actual class of the $k_{th}$ instance. 
The time-sensitive error is defined as 
$$ error(i_k,j_k) \textnormal{=}  e_{i_k}^T 
\begin{bmatrix}
0 & t_{1,k}-t_{2,k} & t_{1,k}-t_{3,k}\\
t_{2,k}-t_{1,k} & 0 & t_{2,k}-t_{3,k}\\
t_{3,k}-t_{1,k} & t_{3,k}-t_{2,k} & 0\\
\end{bmatrix}
e_{j_k} $$
where $e_{i_k}$ and $e_{j_k}$ are one-hot vectors of length $3$ with the $i_k$, $j_k$ position being $1$. 
$t_{1,k}$, $t_{2,k}$ and $t_{3,k}$ are the running time of $\qone,\qtwo,\qthree$ on the $k_{th}$ instance respectively. 
The time-sensitive error measures the gap between the predicted query algorithm and the optimal query algorithm.
It is averaged over all instances in the validation set and across $5$ iterations of cross-validation.
Then, the hyper-parameter setting with the lowest time-sensitive error should be chosen.  
In this way, we are more prone to find the parameter settings that allow us to minimize the query time instead of merely correctly classify each instance. 

Note that, training a feed-forward neural network (lines 2 - 7) take significantly less time compared to the \grating construction process (line 1) since only $N$ ($N \leq 5\%$ of the total number of possible queries) random queries are used. 

\begin{example}
On dataset \texttt{DBpedia-starring}, given $\alpha$=$2$, $\beta$=8, and $\tau$=$9$.
$\qone$ takes $0.53$ seconds to find the queried subgraph.
$\qtwo$ and $\qthree$ takes $0.03$ and $0.05$ respectively. The optimal query processing algorithm on this instance is $\qtwo$. 
Accuracy as a performance metric would give equal penalty to mis-classifying $\qone$ and $\qthree$ as the best algorithm, which is clearly inappropriate. 
Instead, the time-sensitive error gives penalty of $0.02$ if we predict $\qtwo$ and $0.51$ if we predict $\qthree$. 
\end{example} 
\section{Experiments}

In this section, we first validate the effectiveness of the \longname model. 
Then, we evaluate the performance of the index construction algorithms as well as the query processing algorithms.

\subsection{Experiments setting}

\noindent
{\bf Algorithms.} Our  empirical  studies  are  conducted  against  the following algorithms:

\noindent
{\em $\bullet$ Index construction algorithms.}
We compare two $\izero$ construction algorithms: the naive decomposition algorithm \decompnaive and the decomposition algorithm with optimizations \decompopt. We also evaluate the index construction algorithms of $\itwo$ and $\ithree$.
As for $\ione$, we report its size and build time by running the index construction algorithm in \cite{liu2020efficient}. 

\noindent
{\em $\bullet$ Query processing algorithms.}
We use the online computation algorithm presented in Section 4 as the baseline method, denoted as $\qbs$.
We compare it to the index-based query processing algorithms $\qzero, \qone, \qtwo$, and $\qthree$, which are based on $\izero,\ione,\itwo$, and $\ithree$ respectively. 
We also evaluate the hybrid computation algorithm $Q_{hb}$, which depends on a well-trained classifier and the indexes $\qone, \qtwo$, and $\qthree$. 

All algorithms are implemented in C++ and the experiments are run on a Linux server with Intel Xeon E3-1231 processors and $16$GB main memory. \textit{We end an algorithm if the running time exceeds two hours.}

\begin{table*}[tbh]
\centering
\scalebox{1.0}{
\begin{tabular}{cccc|ccccccc}
\noalign{\hrule height 1.23pt}
Dataset &  $|E|$ &  $|U|$ & $|L|$  & $\alpha_{max}$  & $\beta_{max}$  & $\tau_{max}$ & $\delta$ \\ 
\noalign{\hrule height 0.7pt}
Cond-mat (AC)  & $58$K & $38$K & $16$K & $37$ & $13$ & $63$ & $8$ \\
Writers  (WR) & $144$K & $135$K & $89$K & $11$ & $82$ & $99$ & $6$ \\
Producers (PR) & $207$K & $187$K & $48$K & $220$ & $18$ & $219$ & $6$ \\
Movies (ST)  & $281$K & $157$K & $76$K  & $19$ & $215$ & $222$   & $7$ \\
Location (LO)  & $294$K & $225$K & $172$K & $12$ & $853$ & $852$  & $8$ \\
BookCrossing (BX) & $434$K & $264$K & $78$K & $376$ & $100$ & $375$   & $13$ \\
Teams (TM) & $1.4$M & $935$K & $901$K & $11$ & $1063$ & $373$   & $9$ \\
Wiki-en (WC) & $3.80$M & $2.04$M & $1.85$M & $39$ & $7659$ & $7658$  & $18$ \\
Amazon (AZ) & $5.74$M & $3.38$M & $2.15$M & $659$ & $294$ & $658$   & $26$ \\
DBLP (DB) & $8.6$M & $5.4$M & $1.4$M & $421$ & $64$ &  $420$  & $10$ \\
\noalign{\hrule height 1.23pt}
\end{tabular}
}
\vspace{-2mm}
\caption{This table reports the basic statistics of $10$ real graph datasets.}
\label{tab:datainfo}
\end{table*}

\noindent
{\bf Datasets.}
We use $10$ real graphs in our experiments, which are obtained from the website KONECT \footnote{\url{http://konect.uni-koblenz.de/networks/}}. 
Table \ref{tab:datainfo} includes the statistics of these datasets, sorted by the number of edges in ascending order. 
The abbreviations of dataset names are listed in parentheses.
$|E|$ is the number of edges in the graph. 
$|U|$ and $|L|$ are the number of vertices in the upper and lower levels. 
$\alpha_{max}$ is the largest $\alpha$ such that $(\alpha,1)_1$-core exists. 
$\beta_{max}$ is the largest $\beta$ such that $(1,\beta)_1$-core exists. 
$\tau_{max}$ is the largest $\tau$ such that $(1,1)_{\tau}$-core exists. 

\subsection{Effectiveness Evaluation}
\begin{figure}[htb]
\centering
\scalebox{0.8}{
\includegraphics[width=0.46\textwidth]{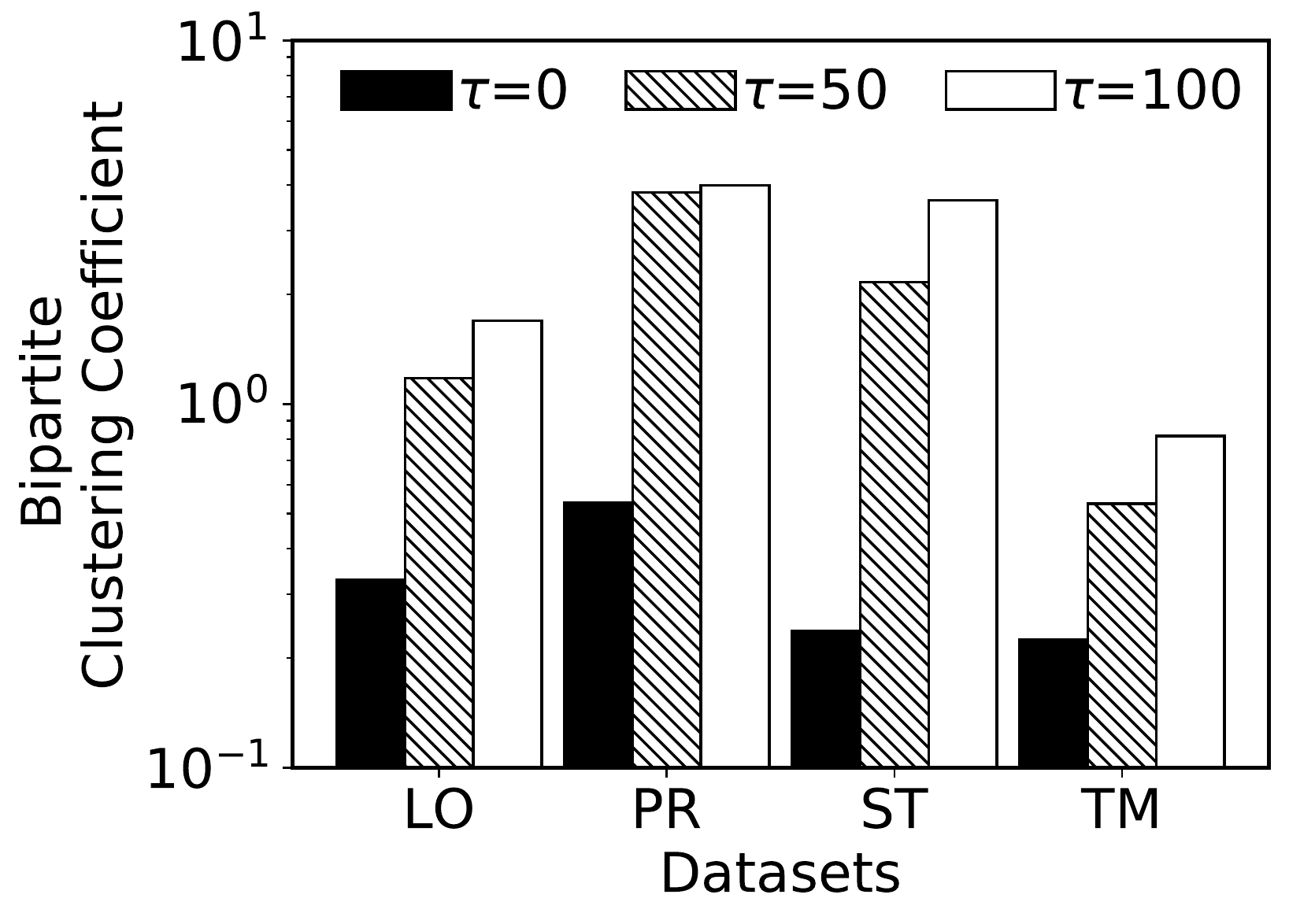}
\includegraphics[width=0.44\textwidth]{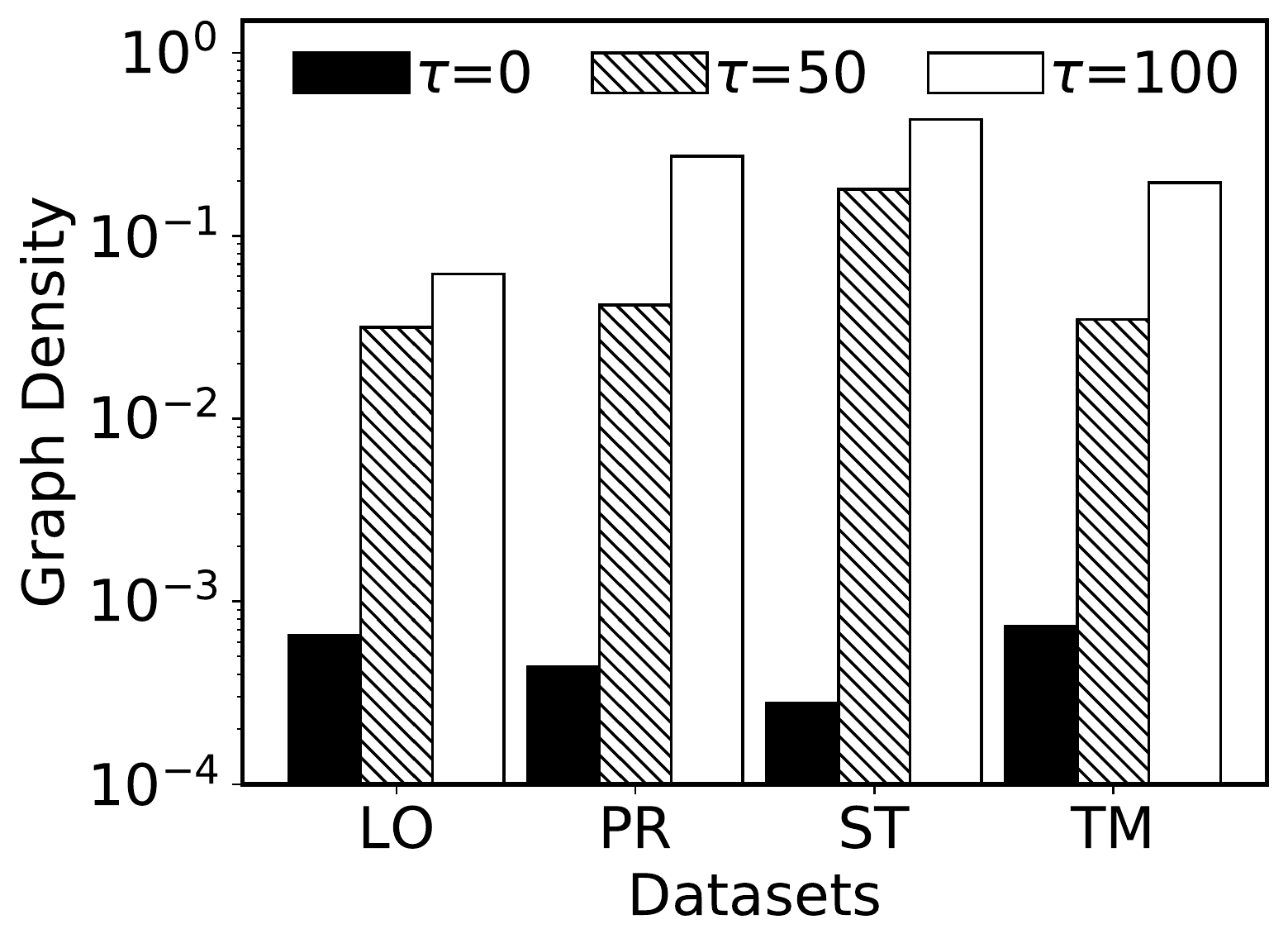}
}
\caption{The cohesive metrics comparisons}
\label{fig:metric}
\end{figure}

In this section, we validate the effectiveness of the \longname model. 
First, we compute some cohesive metrics for \abcore and \shortname. 
Then, we conduct a case study on dataset \texttt{DBLP-2019}. 

\noindent
{\bf Compare \abcore with \longname.}
We compare the graph density and bipartite clustering coefficient for \abcore and \shortname. 
The graph density \cite{sariyuce2018peeling} of a bipartite graph is calculated as $|E|/(|U|\times|L|)$, where $|E|$ is the number of edges and $|U|$ and $|L|$ are the number of upper and lower vertices. 
The bipartite clustering coefficient \cite{aksoy2017measuring} is a cohesive measurement of bipartite networks, which is calculated as $4 \times \btf_G$/$\cate_G$ where $\cate_G$ and $\btf_G$ are the number of caterpillars (three-path) and the number of butterflies in graph $G$ respectively. 
In Figure \ref{fig:metric}, the black bars with $\tau$=$0$, represents the \abcore. 
The shaded bars and the white bars represent the \shortname with $\tau$ being $50$ and $100$ respectively. 
The engagement constraints $\alpha$ and $\beta$ are set to $0.6 \delta$ and $0.4 \delta$ respectively, where $\delta$ is the graph degeneracy. As we can see, on all four datasets, the \shortname has a higher density and bipartite coefficient than \abcore. As $\tau$ increases, both of the metrics increase as well. 
This means that with higher values of $\tau$, we can find subgraphs within \abcore of higher density and cohesiveness. 

\begin{figure}[h!]
\centering  
\includegraphics[width=0.7\textwidth]{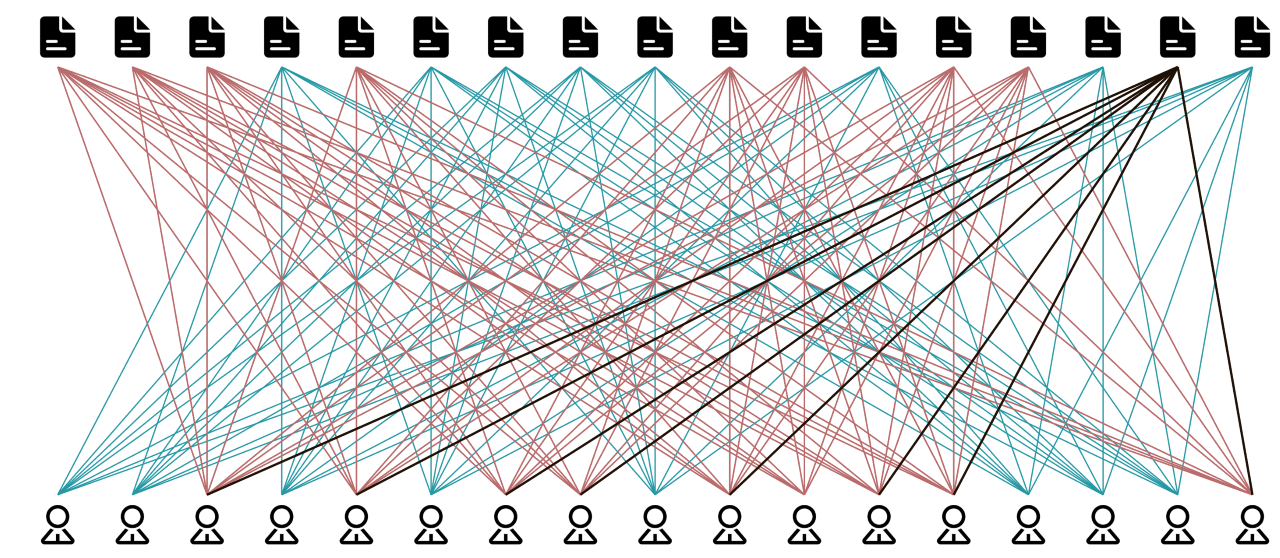}
\vspace{-2mm}
\caption{Case study on DBLP-2019} 
\label{fig:case}
\end{figure}

\noindent
{\bf Case study.}
The effectiveness of our model is evaluated through a case study on the DBLP-2019 dataset. The graph in Fig \ref{fig:case} is an \abcore ($\alpha\textnormal{=}7,\beta\textnormal{=}8$). Given $\tau$=$50$, \shortname excludes the relatively sparse group represented by the light blue lines. The \textit{k}-bitruss ($k$=$56$) represented by the red lines is in \shortname. The black lines are the edges included in \shortname but not in \textit{k}-bitruss. The \shortname and \textit{k}-bitruss involve the same authors, but \textit{k}-bitruss removes the second last paper on the upper level to enforce the tie strength constraint.
Figure \ref{fig:case} implies that:
(1) Although \abcore models vertex engagement via degrees, it fails to distinguish between edges with different tie strength. 
(2) \textit{k}-bitruss models tie strength via butterfly counting, but it forcefully excludes the weak ties between strongly engaged nodes, which leads to the imprecise estimation of tie strength and failure to include important nodes and their incident edges. 
(3) \shortname considers both vertex engagement and tie strength. Its flexibility allows it to capture unique structures that better resemble the communities in reality. 

\begin{figure}[htb]
\centering
\includegraphics[width=0.6\textwidth]{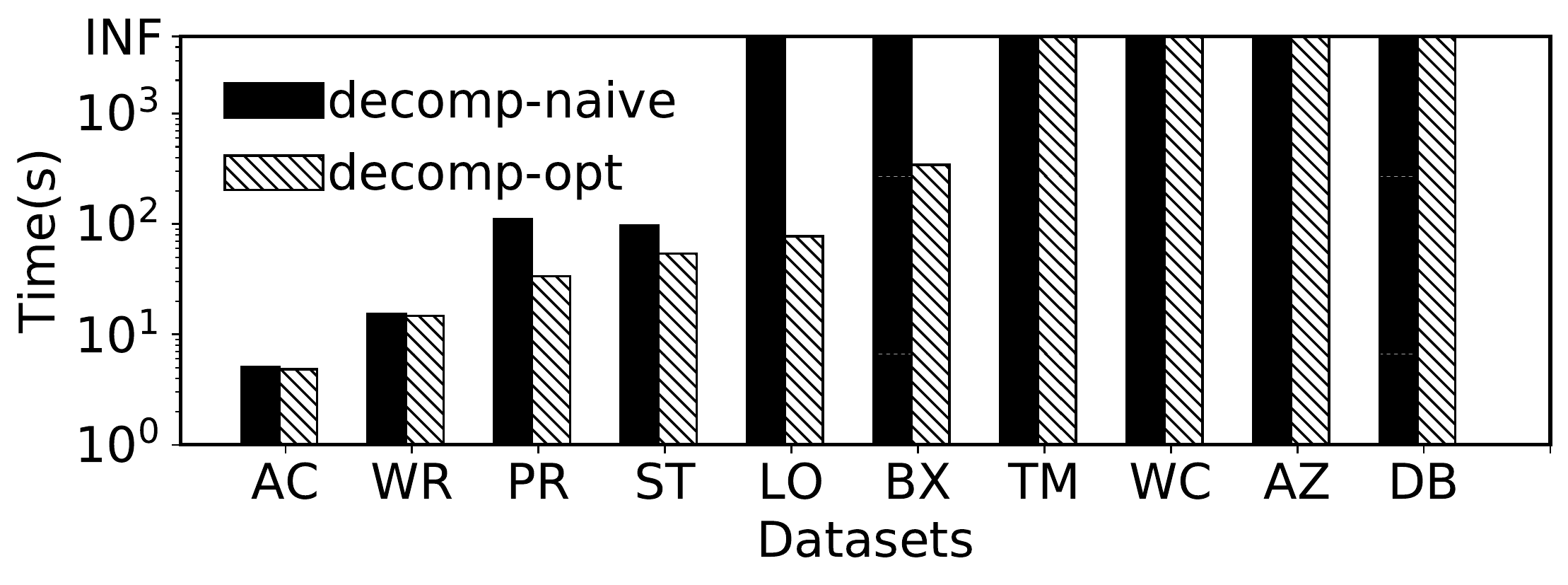}
\caption{The $\izero$ construction time}
\label{fig:index.compare}
\end{figure}
\begin{table*}[tbh]
\centering
\scalebox{1.0}{
\begin{tabular}{c|cccc|cccc}
\noalign{\hrule height 1.23pt}
\multirow{2}*{Data} & \multicolumn{4}{c|}{Index size (MB)} & \multicolumn{4}{c}{Index construction time (sec)} \\
\cline{2-9} 
~ & $\izero$ & $\ione$ &  $\itwo$ & $\ithree$ & $\izero$ & $\ione$ &  $\itwo$ & $\ithree$  \\
\noalign{\hrule height 0.7pt}
Cond-mat  & $1.29$ & $0.78$ & $0.26$ & $0.50$ & $5.11$ & $0.11$ & $0.68$ & $1.37$ \\ 
Writers   & $2.14$ & $2.24$ & $0.93$ & $0.24$ & $18.81$ & $0.24$ & $5.01$ & $1.31$ \\
Producers & $5.55$ & $3.16$ & $0.37$ & $2.43$ & $79.43$ & $0.38$ & $4.99$ & $28.37$ \\
Movies    & $5.26$ & $3.51$ & $2.01$ & $0.53$ & $66.18$ & $0.46$ & $34.29$ & $6.26$ \\ 
Location  & $68.96$ & $4.15$ & $33.22$ & $0.75$ & $77.3978$ & $0.36$ & $49.5368$ & $9.19316$ \\
BookCrossing & $33.56$ & $5.58$ & $3.07$ & $9.32$ & $342.728$ & $1.02$ & $48.8196$ & $75.2869$ \\ 
Teams      & $-$ & $18.42$ & $114.17$ & $2.44$ & time out & $1.94$ & $944.102$ & $127.265$ \\ 
Wiki-en    & $-$ & $46.66$ & $945.91$ & $10.92$ & time out & $9.079$ & $3850.51$  & $680.569$ \\ 
Amazon      & $-$ & $72.96$ & $74.47$ & $129.21$ & time out & $17.184$  & $3598.13$  & $4731.29$ \\
DBLP        & $-$ & $112.60$ & $29.13$ & $159.74$ & time out & $19.44$ & $559.76$ & $3000.08$ \\
\noalign{\hrule height 1.23pt}
\end{tabular}
}
\vspace{-2mm}
\caption{Evaluate the size of indexes and their build time.}
\label{tab:index_size}
\end{table*}

\subsection{Performance Evaluation}
In this part, we evaluate the efficiency of the index construction algorithms and explore the appropriate hyperparameter settings for the feed-forward neural network that $Q_{hb}$ depends on.
Then, we evaluate the efficiency of the query processing algorithms to retrieve \shortname. 

\noindent
{\bf Index construction.} 
First, We compare the build time of $\izero,\ione,\itwo$ and $\ithree$ on all datasets, as reported in Table \ref{tab:index_size}.
The reported build time corresponds to the index construction algorithms with the optimization techniques in Section $6$. 
As shown in \ref{fig:index.compare}, although the computation-sharing and the \texttt{Bloom-Edge}-index based optimizations effectively reduce the running time, the $\izero$ still cannot be built within time limit on \texttt{Teams}, \texttt{Wiki-en}, \texttt{Amazon} or \texttt{DBLP}. 
This is because $\izero$ stores all the decomposition results and it takes the longest time to build, followed by $\itwo,\ithree$, and $\ione$.
When the graph is denser on the upper level, $\ithree$  takes longer to construct than $\itwo$.
For example, in \texttt{DBLP}, where $\dmaxv\textnormal{=}119 < \dmaxu\textnormal{=}951$, $\itwo$ is built within $10$ minutes while $\ithree$ is built in $50$ minutes. 
$\ione$ construction is the fastest as it does not involve any butterfly counting or updating of support. 
In addition, we report the index sizes in Table \ref{tab:index_size}. 
The size of $\izero$ is larger than $\ione,\itwo \textnormal{ and } \ithree$. 
In summary, the \gratings that the hybrid computation algorithm depends on are space-efficient and can be built within a reasonable time.\\

\begin{figure}[htb]
\centering
\scalebox{0.8}{
\includegraphics[width=0.46\textwidth]{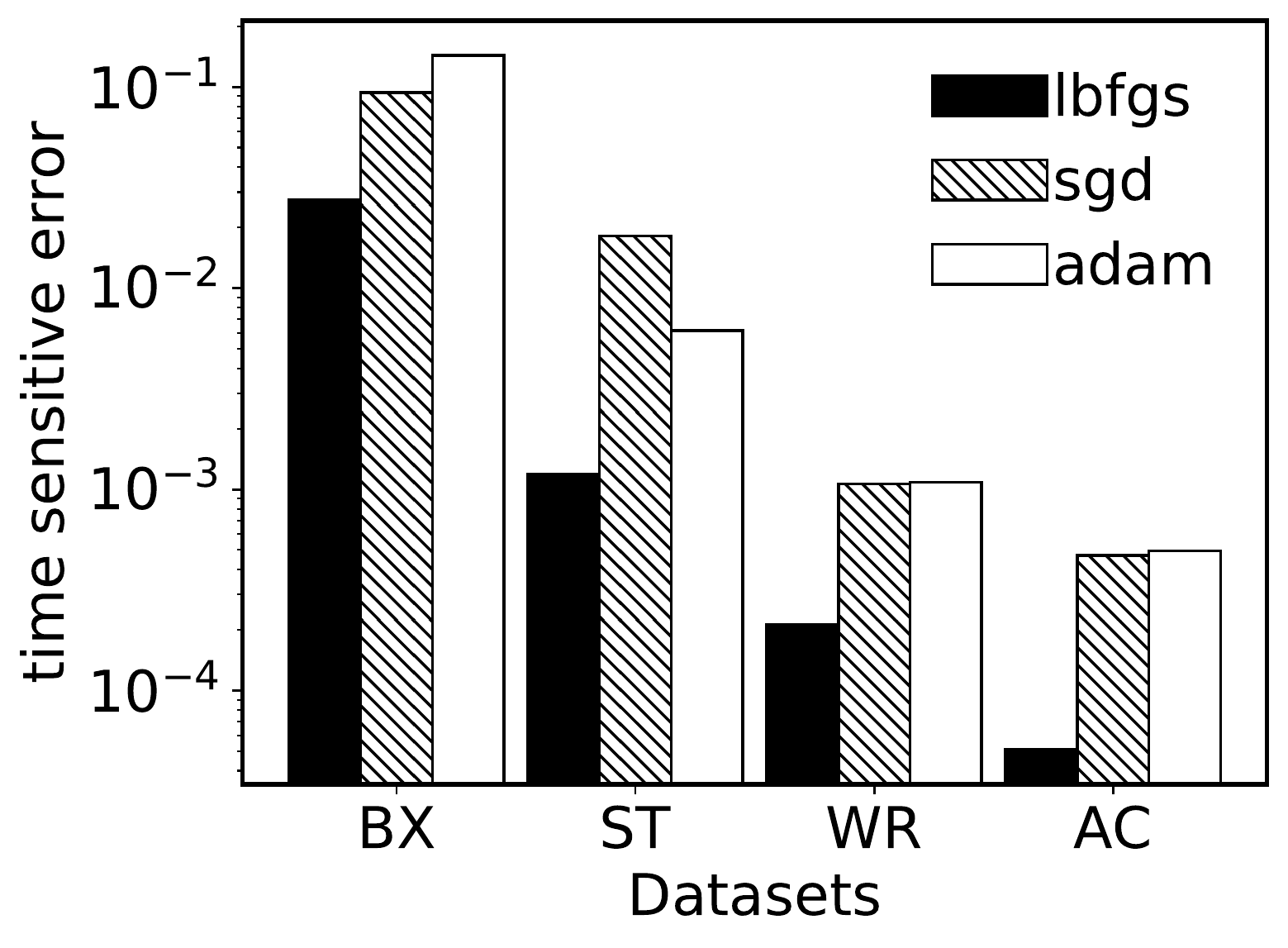}
\includegraphics[width=0.46\textwidth]{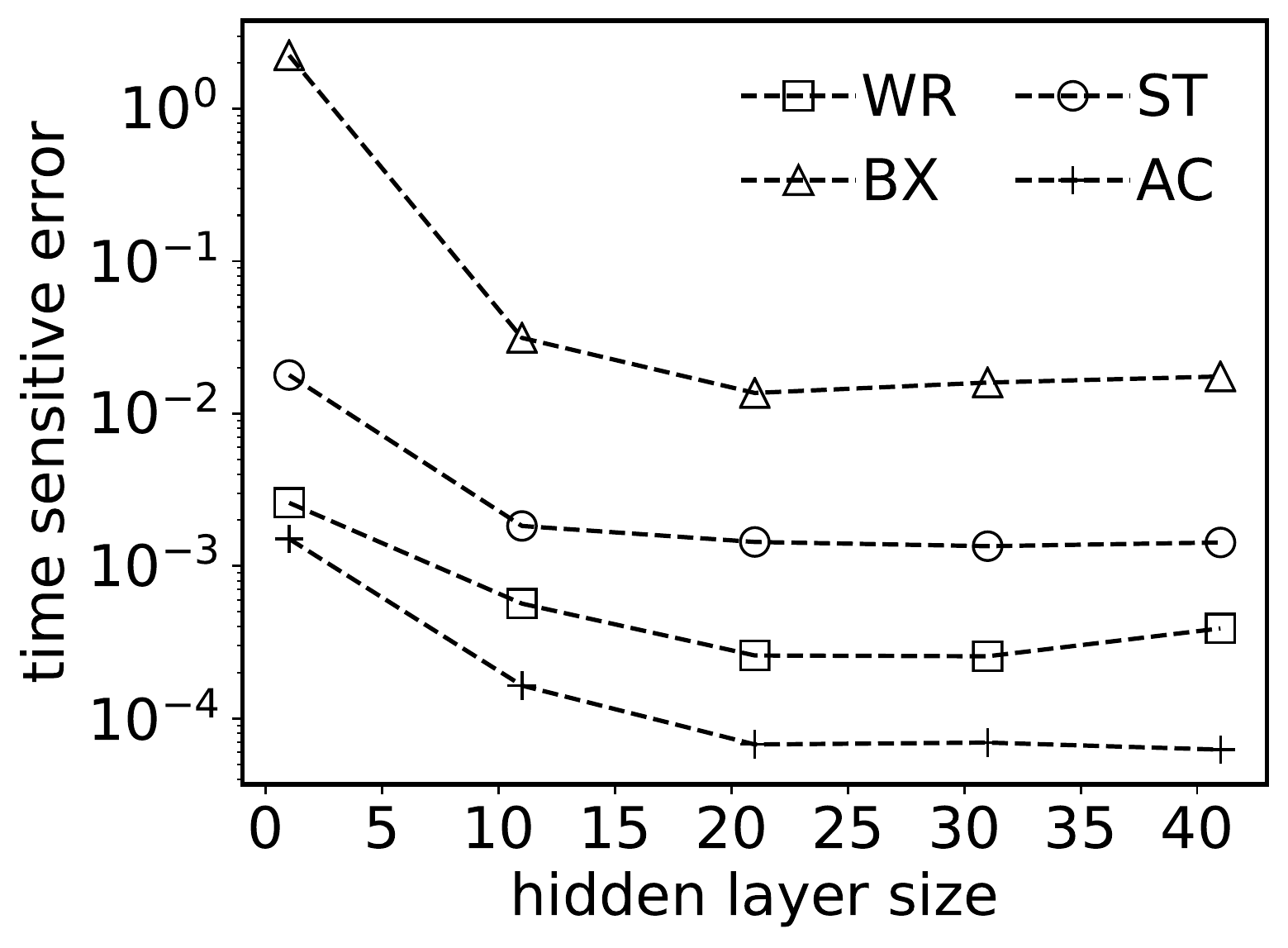}
}
\caption{Effects of hyperparameters}
\label{fig:tune}
\end{figure}

\noindent
{\bf Tuning hyperparameters for the neural network.}
When training the neural network for the hybrid computation algorithm, we choose the hyperparameters (the type of optimizer and size of the hidden layer) that minimizes the time-sensitive error from cross-validation.
For each graph $G$, we set the size of hidden layer to $50$ and test the \textit{stochastic gradient descent}, \textit{L-BFGS method} and \textit{Adam} and compare the time-sensitive error. 
As shown in Figure  \ref{fig:tune}, the L-BFGS method consistently 
outperforms the other methods and thus is chosen in our model. 
Then, we explore the effect of the size of the hidden layer on our model. 
We report the change of time-sensitive error w.r.t varying hidden layer size on dataset \texttt{DBpedia-location} and the trends are similar on other datasets.
As shown in the plot, $30$ hidden units are enough for the classifier built on most tested datasets and beyond this point, more hidden units have little effect on the performance of the model. 

\begin{table*}[t!]
\centering
\scalebox{1.0}{
\begin{tabular}{ccccccc}
\noalign{\hrule height 1.23pt}
Dataset  & $\qbs$ & $\qzero$ & $\qone$ &  $\qtwo$ & $\qthree$ & $Q_{hb}$  \\ 
\noalign{\hrule height 0.7pt}
Cond-mat &  $0.139$ & $0.004$ & $0.030$ & $0.008$ & $0.009$ & $0.006$ \\
Writers  & $0.406$ & $0.011$ & $0.069$ & $0.011$ & $0.009$ & $0.006$ \\
Producers &  $0.560$ & $0.022$ & $0.200$ & $0.047$ & $0.037$ & $0.029$ \\
Movies  & $0.871$ & $0.023$ & $0.235$ & $0.028$ & $0.041$ & $0.027$ \\
Location & $11.782$ & $0.285$ & $7.005$ & $1.755$ & $2.895$ & $0.234$\\
BookCrossing & $44.397$ & $0.147$ & $13.271$ & $2.935$ & $4.794$ & $1.722$ \\
Teams &  $59.510$ & $-$ & $10.911$ & $3.080$ & $2.778$ & $1.394$ \\
Wiki-en  & $128.536$ & $-$ & $28.589$ & $2.638$ & $13.613$ & $1.775$\\
Amazon  & $973.026$ & $-$ & $153.085$ & $88.673$ & $59.228$ & $16.432$ \\
DBLP & $61.101$ & $-$ &  $1.843$ & $1.321$ & $1.055$ & $0.269$ \\
\noalign{\hrule height 1.23pt}
\end{tabular}
}
\vspace{-2mm}
\caption{This table reports the average response time for all query processing algorithms.}
\label{tab:summary}
\end{table*}
\begin{figure}[h!]
\centering  
\includegraphics[trim= -40 0 0 0 ,width=0.8\textwidth]{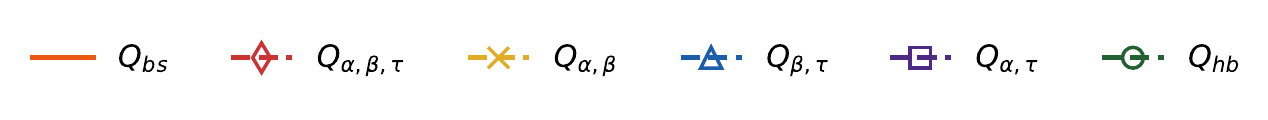}

\subfigure[LO (Vary $\alpha$)]{
\label{Fig.pr.a}
\includegraphics[width=0.3\textwidth]{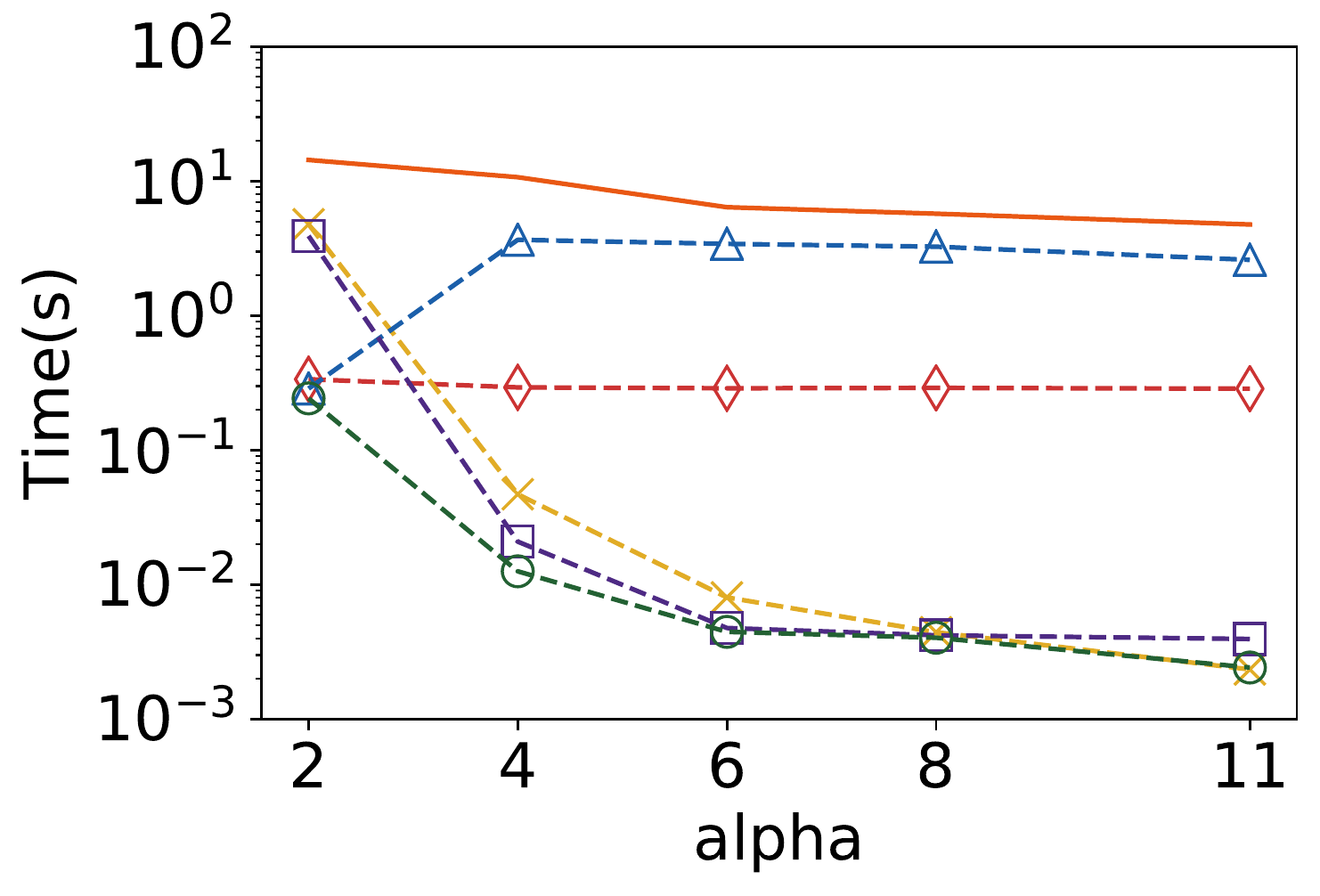}}
\subfigure[LO (Vary $\beta$)]{
\label{Fig.pr.b}
\includegraphics[width=0.3\textwidth]{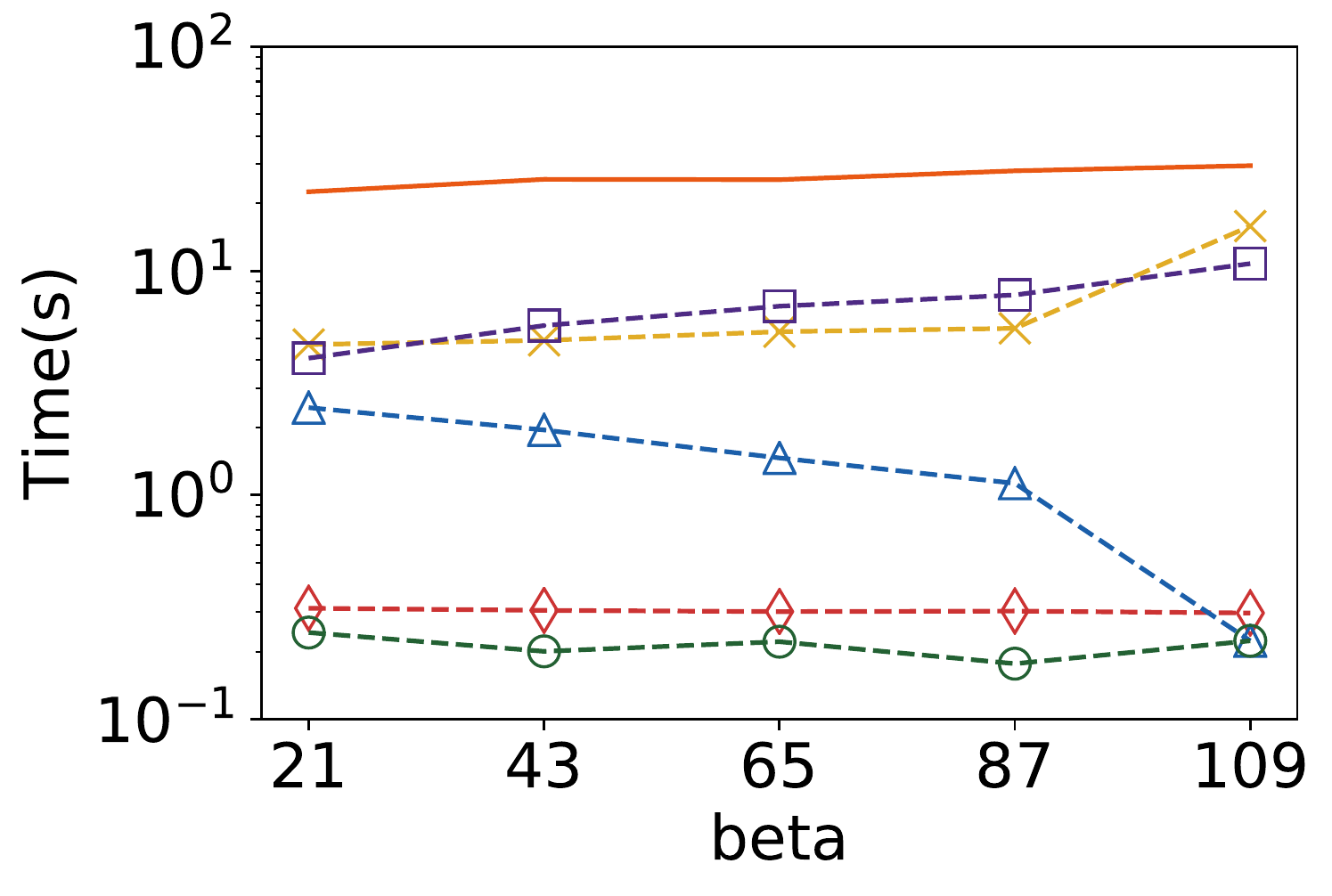}}
\subfigure[LO (Vary $\tau$)]{
\label{Fig.pr.g}
\includegraphics[width=0.3\textwidth]{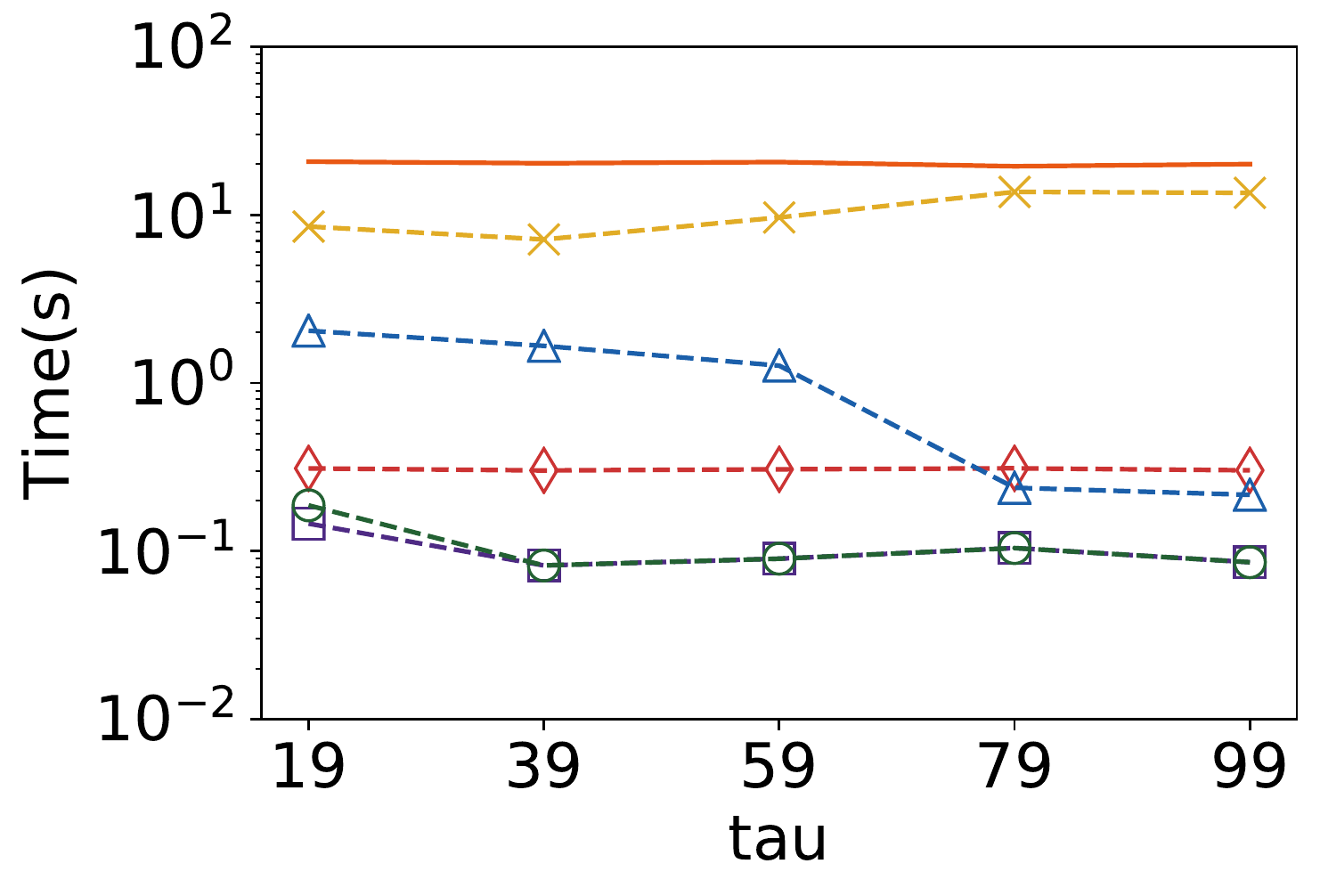}}

\subfigure[BX (Vary $\alpha$)]{
\label{Fig.pa.a}
\includegraphics[width=0.3\textwidth]{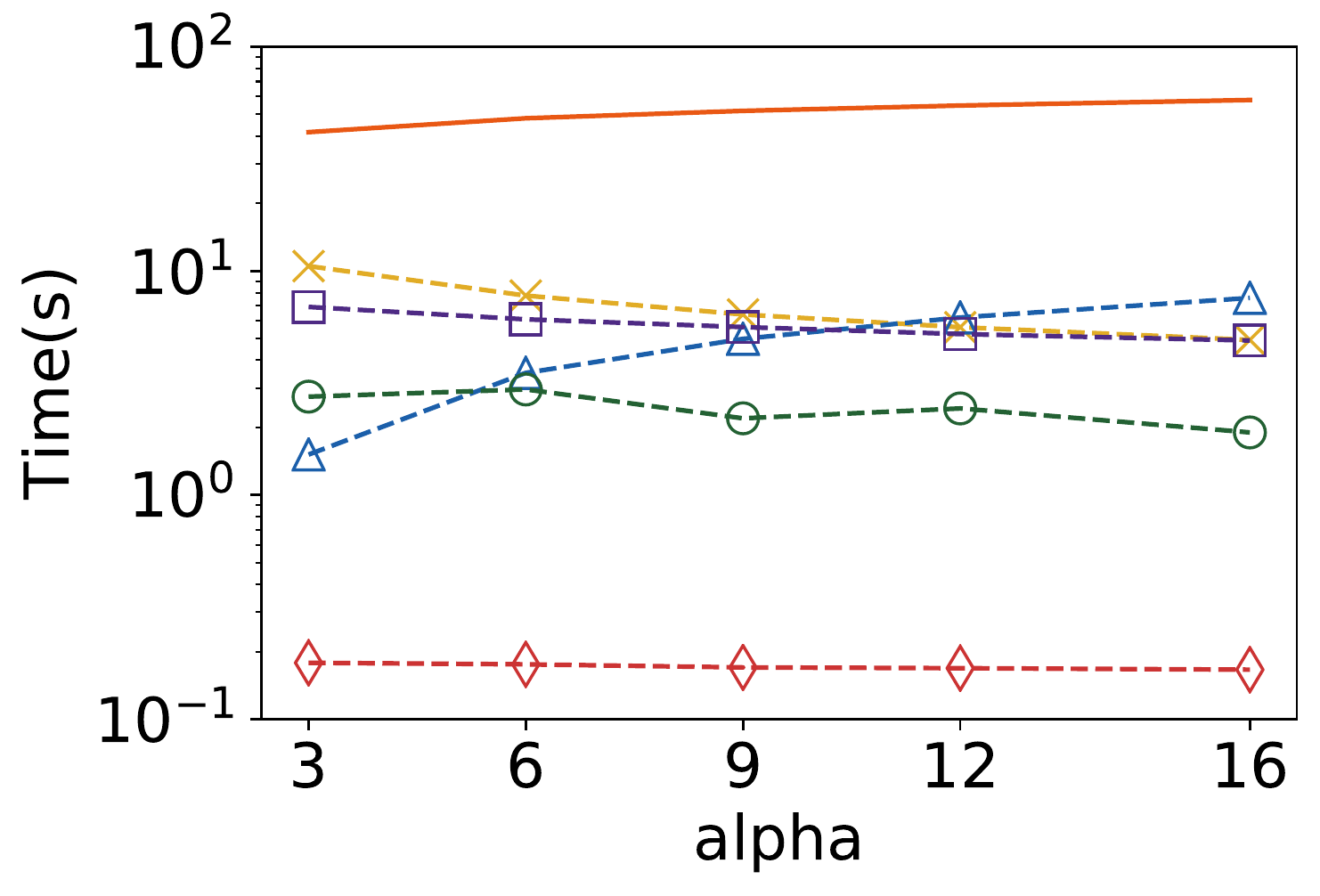}}
\subfigure[BX (Vary $\beta$)]{
\label{Fig.pa.b}
\includegraphics[width=0.3\textwidth]{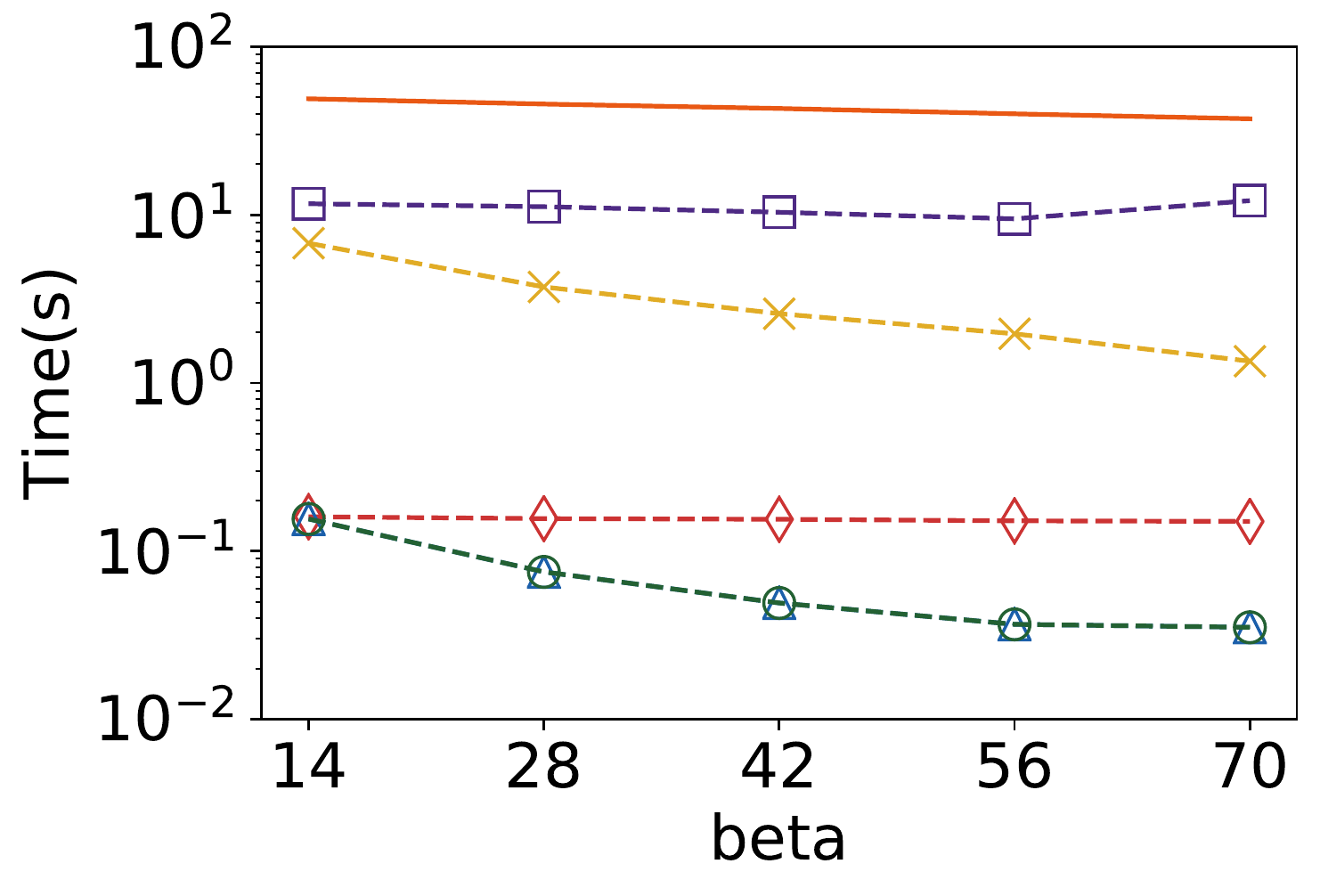}}
\subfigure[BX (Vary $\tau$)]{
\label{Fig.pa.g}
\includegraphics[width=0.3\textwidth]{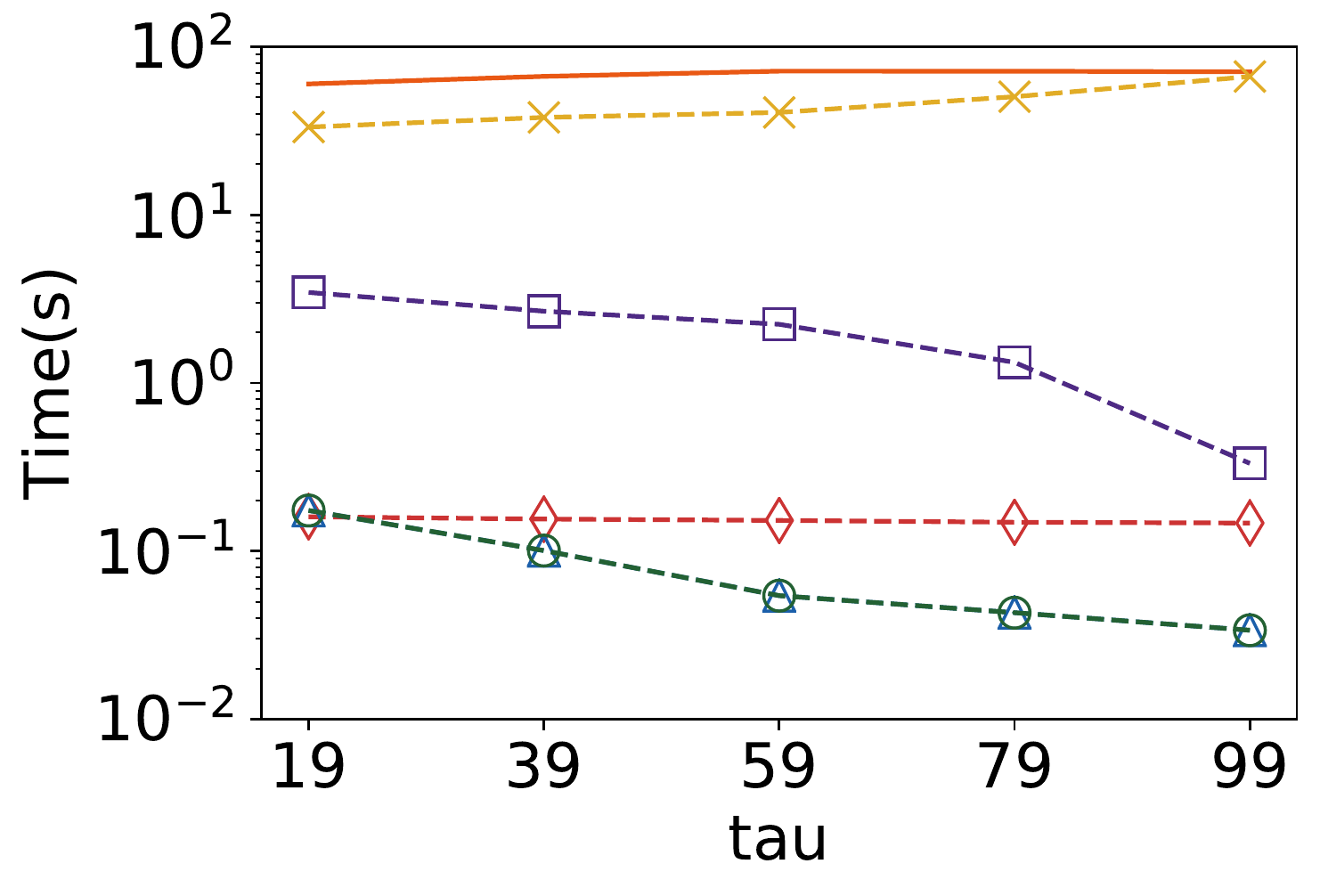}}
\vspace{-2mm}
\caption{Effects of varying parameters on each query processing algorithm.}
\label{Fig.queries}
\end{figure}

\noindent
{\bf Average query time of $\qzero, \qone, \qtwo, \qthree$ and $Q_{hb}$.} 
For each algorithm, we report the average response time of $50$ randomly generated queries on each dataset in Table \ref{tab:summary}.
As expected, all index based algorithms outperform $\qbs$, as $\qbs$ computes \shortname from the input graph.
Among them, $\qzero$ performs the best on most datasets, as it fetches the vertices from $\izero$ in optimal time and then restores the edges. 
However, the long build time of $\izero$ makes $\qzero$ not scalable to larger graphs like \texttt{Team},\texttt{Wiki-en},\texttt{Amazon} or \texttt{DBLP}. 
The performances of $\qone,\qtwo$, and $\qthree$ differ a lot from each other and across datasets. 
On average, $\qone$ is slower than $\qtwo$ and $\qthree$, because it needs to delete many edges from \abcore to get \shortname especially when $\tau$ is large. 
As $Q_{hb}$ is trained to pick the fastest from $\{ \qone,\qtwo,\qthree \}$, it outperforms these $3$ algorithms on average in all datasets.
In addition, $Q_{hb}$ outperforms the online computation algorithm $\qbs$ by up to two orders of magnitude.

\noindent
{\bf Evaluate the effect of $\alpha,\beta$ and $\tau$.} 
We investigate the effects of varying $\alpha, \beta, \tau$ on each query processing algorithm. 
We input three types of query streams for all algorithms, each of which increments one of $\alpha$, $\beta$, and $\tau$ while letting the other two parameters be generated randomly. 
As trends are similar, we only report the results on \texttt{Location} and \texttt{BookCrossing} in Figure \ref{Fig.queries}.
Each data point in the figure represents the average response time of $50$ random queries.
As expected, index-based query processing algorithms always outperform $\qbs$.
The performance of $\qzero$ is not affected much by the varying parameters, while the \grating based algorithms are highly sensitive to them. 
Each of $\qone,\qtwo$, and $\qthree$ tends to perform better when the increased parameter results in a smaller subgraph in the index. 
For instance, $\qtwo$ performs better as $\beta$ or $\tau$ increases. 
In contrast, the hybrid computation algorithm $Q_{hb}$ has very stable performance, as it stays close to and in many cases outperforms the fastest of $\qone,\qtwo$, and $\qthree$. 
In summary, the hybrid computation algorithm $Q_{hb}$ with a well-trained classifier can adjust its querying processing algorithm to different parameters and datasets.
\section{Conclusion}
In this paper, we introduce a novel cohesive subgraph model, \longname, which is the first to consider both tie strength and vertex engagement on bipartite graphs.
We propose a decomposition-based index $\izero$ that can retrieve vertices of any \shortname in optimal time. 
We also apply computation sharing and \beindex-based optimizations to speed up the index construction process of $\izero$. 
To balance space-efficient index construction and time-efficient query processing, we propose a learning-based hybrid computation paradigm.
Under this paradigm, we introduce three \gratings and train a feed-forward neural network to predict which index is the best choice to process an incoming \shortname query. 
The efficiency of the proposed algorithms and the effectiveness of our model are verified through extensive experiments. 


\end{document}